\documentclass[11pt]{article}

\usepackage[letterpaper, hmargin=1in, top=1in, bottom=1in,
footskip=0.6in]{geometry}
\usepackage[T1]{fontenc}

\usepackage{titlesec}
\titleformat{\subsection}[runin]{\large\bfseries}{\thesubsection.}{.5em}{}[.]
\titlespacing{\subsection}{0pt}{2ex plus .1ex minus .2ex}{.8em}
\titleformat{\subsubsection}[runin]{\normalfont\bfseries}{\thesubsubsection.}{.3em}{}[.]
\titlespacing{\subsubsection}{0pt}{2ex plus .1ex minus .2ex}{.5em}
\titleformat{\paragraph}[runin]{\normalfont\itshape}{\theparagraph.}{.3em}{}[.]
\titlespacing{\paragraph}{0pt}{1ex plus .1ex minus .2ex}{.5em}

\usepackage{amsmath}
\usepackage{amssymb}
\usepackage{amsfonts}
\usepackage{latexsym}
\usepackage{amsthm}
\usepackage{amsxtra}
\usepackage{amscd}
\usepackage{bm}
\usepackage{tensor}
\usepackage{microtype}

\usepackage{graphicx,xcolor}
\usepackage{float}
\definecolor{darkred}{rgb}{0.9,0,0.3}
\definecolor{darkblue}{rgb}{0,0.3,0.9}

\definecolor{vdarkred}{rgb}{0.7,0,0.2}
\definecolor{vdarkblue}{rgb}{0,0.2,0.7}

\usepackage[nottoc,notlof,notlot]{tocbibind}
\usepackage{cite}
\usepackage{tikz}
\usetikzlibrary{arrows.meta,positioning,calc}
\usepackage[colorlinks,linkcolor=vdarkblue,urlcolor=black,
citecolor=vdarkred]{hyperref}
\usepackage{cleveref}

\allowdisplaybreaks
\numberwithin{equation}{section}
\numberwithin{figure}{section}
\titleformat*{\section}{\Large\bfseries}

\renewcommand{\leq}{\leqslant}

\renewcommand{\geq}{\geqslant}

\renewcommand{\epsilon}{\varepsilon}

\theoremstyle{plain}
\newtheorem{theorem}{Theorem}[section]
\newtheorem*{theorem*}{Theorem}
\newtheorem{lemma}[theorem]{Lemma}
\newtheorem*{lemma*}{Lemma}
\newtheorem{corollary}[theorem]{Corollary}
\newtheorem*{corollary*}{Corollary}
\newtheorem{proposition}[theorem]{Proposition}
\newtheorem*{proposition*}{Proposition}

\newtheorem*{conjecture*}{Conjecture}

\theoremstyle{definition}

\newtheorem*{definition*}{Definition}

\newtheorem*{example*}{Example}
\newtheorem{remark}[theorem]{Remark}
\newtheorem*{remark*}{Remark}

\newtheorem*{assumption*}{Assumption}

\renewcommand{\b}[1]{\boldsymbol{\mathrm{#1}}}

\newcommand{\bb}{\mathbb}
\renewcommand{\cal}{\mathcal}

\newcommand{\ol}[1]{\overline{#1} \!\,}
\newcommand{\wh}{\widehat}
\newcommand{\wt}{\widetilde}

\newcommand{\E}{\mathbb{E}}

\newcommand{\e}{\mathrm{e}}

\newcommand{\ii}{\mathrm{i}}
\newcommand{\dd}{\mathrm{d}}

\newcommand*{\deq}{\mathrel{\vcenter{\baselineskip0.65ex \lineskiplimit0pt
			\hbox{.}\hbox{.}}}=}
\newcommand*{\eqd}{=\mathrel{\vcenter{\baselineskip0.65ex \lineskiplimit0pt
			\hbox{.}\hbox{.}}}}

\newcommand{\floor}[1]{\lfloor {#1} \rfloor}
\newcommand{\ceil}[1]{\lceil {#1} \rceil}
\newcommand{\qq}[1]{[\![{#1}]\!]}

\DeclareMathOperator{\tr}{Tr}
\DeclareMathOperator{\var}{Var}

\DeclareMathOperator{\re}{Re}

\DeclareMathOperator{\sgn}{sgn}

\newcommand{\fs}{ {\mathfrak s} }

\newcommand{\ee}{\varepsilon}

\makeatletter
\newcommand*{\rom}[1]{\expandafter\@slowromancap\romannumeral #1@}
\makeatother

\newcommand{\trn}{\operatorname{tr}}

\definecolor{prooflayerblue}{RGB}{100,162,235}
\definecolor{prooflayerpurple}{RGB}{164,96,244}
\definecolor{prooflayergold}{RGB}{192,153,59}

\title{\bf \Large Free energy and spectral edge of the SYK model \vspace{0.5em}}

\author{Yukun He\footnote{Shanghai Center for Mathematical Sciences and School of Mathematical Sciences, Fudan University. Email: \href{mailto:heyukun@fudan.edu.cn}{heyukun@fudan.edu.cn}.}\vspace{1em}}

\begin{document}
\maketitle
\begin{abstract}
Consider the Sachdev--Ye--Kitaev model with $N$ Majorana fermions and interaction order $q$. For even $4\leq q\ll \sqrt{N}$, we present a microscopic framework that determines the free energy limit at every fixed inverse temperature $\beta>0$. In particular, this gives the first rigorous derivation of the Schwinger--Dyson thermodynamics of the model. For all even $q\in \{4,6,...,N-4\}$, we further determine the asymptotic location of the largest eigenvalue, which answers a question posed by
Feng--Tian--Wei \cite{FengTianWei1}. Consequently, we obtained the free energy limit along every sequence $\beta\equiv \beta_N\to\infty$.
\end{abstract}

\section{Introduction}

In this paper, let $N,q$ be even integers such that $4\leq q\leq N-4$. Abbreviate $\qq{N}\deq \{1,2,...,N\}$, and $L\deq2^{N/2}$. We consider the Hamiltonian
\begin{equation} \label{1.1}
	H
	\deq\bigg(\frac{(q-1)!}{2^qN^{q-1}}\bigg)^{1/2}
	\sum_{A\subset\qq{N},\,|A|=q}J_A\Psi_A\,,
\end{equation}
where the $J_A$ are i.i.d.\ standard Gaussian random variables, and
$\Psi_A\in\bb C^{L\times L}$ are the standard Hermitian Majorana
monomials.  Their precise normalization is fixed at the beginning of Section
\ref{sec2}.  In particular, for $q$-sets $A,B\subset\qq{N}$,
\begin{equation} \label{1.2}
\Psi_A^2=I\,,\quad
\Psi_A\Psi_B=(-1)^{|A\cap B|}\Psi_B\Psi_A\,.
\end{equation}
This is the (thermodynamically normalized) Sachdev--Ye--Kitaev (SYK) model
\cite{SY93,K15}, a quantum many-body system that has become one of the
most studied models in theoretical physics over the past decade.
Representative theoretical results include the global limiting spectrum and
order bounds on the spectral edge, a central limit theorem for linear
statistics, and concentration of the empirical spectral measure
\cite{FengTianWei1,FengTianWei2,FengTianWei3}, as well as finite-$N$
spectral-moment expansions \cite{GJV18,JV18}, spectral and thermodynamic
analyses \cite{GGV16,GGV17,PR16}, operator-growth bounds
\cite{Lucas20,ChenLucas21}, and a high-temperature free-energy limit
\cite{GPSZ26}.  Much of this interest stems from its surprising connections
among condensed matter physics, quantum chaos, and quantum gravity
\cite{MS16,CAHPSSSST17,CGPS22,BAK16,KS18}.
Georges--Parcollet--Sachdev \cite{GPS01} computed the nonzero
zero-temperature entropy of the large-$N$ fermionic spin-liquid solution
in a related quantum Heisenberg spin-glass model.

\subsection{Free energy}
Our main goal here is to compute the SYK free energy at every fixed positive
temperature. Equivalently, we compute the normalized $\log$ partition
function, a fundamental thermodynamic quantity. The
normalized pressure encodes the equilibrium Gibbs state, its $\beta$-derivative gives the internal-energy
density, and its zero-temperature slope controls the extremal spectrum. For fixed $q\geq 4$, the Schwinger–Dyson
formalism predicts this quantity at large N, but a microscopic derivation without a temperature
restriction was previously unavailable.

We write $\tr$ for the ordinary matrix trace and
$\trn$ for the trace normalized by matrix dimension. For instance, we have $\trn H=L^{-1}\tr H$. Our primary thermodynamic observables are the annealed and quenched free-energy densities
\[
\begin{aligned}
	f_{\rm ann}(\beta)
	\deq-(N\beta)^{-1}\log\E\tr\e^{\beta H}\,, \quad 
	f_{\rm que}(\beta)
	\deq-(N\beta)^{-1}\E\log\tr\e^{\beta H}\,.
\end{aligned}
\]
At each fixed inverse temperature, they are equivalent to the normalized pressures
\begin{equation*}
	p_{\rm ann}(\beta)\deq N^{-1}\log \E\trn \e^{\beta H}\,,
	\quad
	p_{\rm que}(\beta)\deq N^{-1}\E\log\trn \e^{\beta H}\,,
\end{equation*} 
which are the natural quantities for the microscopic calculation.

When $q$ is fixed, the free energy limit is described by the Schwinger--Dyson pressure
\begin{equation} \label{pSD}
	 p_{{\rm SD},q}(\beta)\deq\int_0^\beta\frac{u}q
	\int_0^1G_{u,q}(\tau)^q\,\dd \tau\,\dd u\,,
\end{equation}
where $G_{\beta,q}$ is the unique normalized positive Lehmann solution of the
Schwinger--Dyson equation. Its precise definition and construction are given in
Lemma \ref{lem2.4}. For $q=4$, the cubic self-energy relation was first derived by
Sachdev--Ye \cite{SY93} in a random quantum Heisenberg model. We also call the
corresponding self-consistency equations the Sachdev--Ye equations. Throughout,
we use the name Schwinger--Dyson equations, following \cite{MS16,KS18}.
For $1\ll q\ll \sqrt{N}$, the free energy is associated with
\begin{equation} \label{bbb}
b\equiv b_q\deq \sqrt{2^{q-1}/q}\quad \mbox{and} \quad f(\beta)\deq\max_{0\leq v\leq \pi/2} (\beta\sin v-v^2)\,.
\end{equation} 
We may now state our main result.

\begin{theorem}[Free energy]
	\label{mainthm2}
	(i)
	Fix an even integer $q\geq4$.  For every fixed $\beta>0$,
	\begin{equation*}
		\lim_{N\to\infty}p_{\rm ann}(\beta)
		=\lim_{N\to\infty}p_{\rm que}(\beta)
		=p_{{\rm SD},q}(\beta)\,.
	\end{equation*}
Equivalently, the corresponding annealed and quenched free-energy densities
both converge to
\[
-\beta^{-1}\bigg(p_{{\rm SD},q}(\beta)+\frac12\log2\bigg)\,.
\]

(ii) Suppose $1\ll q\ll \sqrt{N}$. Then for every fixed $\beta>0$, we have
\[
\lim_{N\to\infty}q^2p_{\rm ann}(b\beta)
=\lim_{N\to\infty}q^2p_{\rm que}(b\beta)
=f(\beta)\,.
\]

\end{theorem}

\begin{remark} \label{rem1.2}
(i) For fixed even $q\geq 4$, Gamarnik--Pernice--Schmidhuber--Zlokapa \cite{GPSZ26} recently proved the existence
of the free-energy limit for sufficiently small fixed $\beta$ and represented it by a convergent combinatorial
expansion, but left open the analytic identification of this limit with the Schwinger--Dyson pressure.
Their proof expands the trace exponential and encodes repeated interaction labels by a sparse
hypergraph and a pairing diagram. In their range, the tilted expansion is controlled by small,
essentially tree-like components. Their argument does not directly extend to
arbitrary fixed $\beta$: the estimates controlling non-tree components require sufficiently small $\beta$.

(ii) For $q=\Omega(\sqrt{N})$, as the spectrum is more compact, the free energy is easier to obtain. Suppose $q\leq N/2$ and $q^2/N\to \alpha_*$ for some $\alpha_*\in (0,\infty]$. Set $\fs\deq \e^{-2\alpha_*}\in [0,1)$. The $q$-Hermite law \cite{BINT} is
\begin{equation} \label{qHermite}
	\varrho_\fs(x)\deq
	\frac{\sqrt{1-\fs}}{2\pi}\sqrt{(4-(1-\fs)x^2)_+}
	\prod_{k=1}^{\infty}
	(1-\fs^k)\Big((1+\fs^k)^2-(1-\fs)\fs^k x^2\Big)\,,
\end{equation} 
It is easy to see that $\varrho_0(x)=\frac{1}{2\pi}\sqrt{(4-x^2)_+}$ is the semicircle density. Define
$
L_\fs(\beta)\deq \log \int \e^{\beta x} \varrho_{\fs}(x)\dd x
$
for all $\fs\in [0,1)$.
Furthermore, write
\begin{equation} \label{ddd}
	d\equiv d_q\deq\bigg(\frac{(q-1)!}{2^qN^{q-1}}\binom Nq\bigg)^{1/2}\,.
\end{equation}
By the global law of Feng--Tian--Wei \cite[Theorem 1]{FengTianWei1} and \eqref{2.5}, it is not hard to show that, for every fixed $\beta>0$, we have
\[
\lim_{N\to\infty}Np_{\rm ann}(\beta/d)
=\lim_{N\to\infty}Np_{\rm que}(\beta/d)
=L_\fs(\beta)\,.
\]
An analogous transition for quantum spin
glass Hamiltonians was proved by Erd\H{o}s--Schr\"oder \cite{ES14}.
\end{remark}

To prove Theorem \ref{mainthm2}\,(i), we express the pressure derivative through
expectations of products of physical fields, using Stein's Lemma. The main
difficulty is to approximate these expectations by products of two-point
functions. To this end, we introduce auxiliary Majorana matrices, two for each
physical index and each source parameter, with quadratic couplings to the
physical matrices. The variances of site averages of products of physical and
auxiliary matrices are $O(N^{-1})$ after averaging over the source parameters;
see \eqref{2.22a}.
Consequently, the averaged error in factorizing their joint moments tends to
zero.

The auxiliary matrices satisfy linear differential equations. Their explicit
solutions express these moments as time integrals of the physical expectations
against explicit kernels, up to $O(N^{-1})$ errors from repeated site indices.
As the number of source parameters grows, these kernels generate a dense set
of continuous test functions. This gives the time-ordered factorization
\eqref{2.13a}. Stein's Lemma then gives the Schwinger--Dyson equation for the
limiting two-point function, for every fixed $\beta>0$. Integrating the pressure
derivative gives \eqref{pSD}. The quadratic bound in Lemma \ref{lem2.2} allows
us to remove the sources after taking $N\to\infty$ and recover the pressure
of $H$.

To prove Theorem \ref{mainthm2}\,(ii), we use the Gaussian moment expansion
\eqref{2.45}, whose terms contain products of $(-1)^{|A\cap B|}$. The scale
$b\beta$ requires moments of order $N/q^2$, so convergence of moments of fixed
order is not sufficient. A tree estimate and induction compare these moments
with those of the $q$-Hermite law up to a factor $\e^{O(N/q^3)}$. Its logarithm,
multiplied by $q^2/N$, is $O(q^{-1})$. The resulting pressure has the same
limit as the large-$q$ Schwinger--Dyson pressure in Section \ref{secsub}.

\subsection{Spectral edge}
Our second goal is to study the asymptotic location of the largest eigenvalue of $H$. For numerical convenience, write
\begin{equation} \label{1.5}
a\equiv a_q\deq\prod_{j=1}^{q-1}\Big(1-\frac jN\Big)\,,\quad \mathcal H\deq\bigg(\frac{q2^q}{aN}\bigg)^{1/2}\hspace{-0.1cm}H
	=\binom Nq^{-1/2}
	\sum_{A\subset\qq{N},\,|A|=q}J_A\Psi_A\,,
\end{equation}
and $\lambda_1\deq \lambda_{\max}(\cal H)$.

For $q=2$, the model is exactly solvable, and it is known that
$\lambda_1=4\sqrt{2N}/(3\pi)+o(\sqrt N)$
\cite{FengTianWei1,MS16}.  For even $q\geq4$, no explicit diagonalization is available. The Hamiltonian acts on a space of
dimension $2^{N/2}$, but involves only $\binom Nq$ independent random
coefficients, placing the high moments needed to detect the spectral edge
beyond the reach of classical moment methods.  Previously, the rough upper
bound
\[
\E\lambda_1\leq\sqrt{N\log2}
\]
was obtained in \cite{FengTianWei1}.  Anschuetz--Chen--Kiani--King proved a
high-probability lower bound of order $\sqrt N/q$ and a quantitative
annealed--quenched free-energy comparison, without identifying the common
limiting pressure \cite{ACKK25}.  Moreover, \cite{MS16} gave a formal
path-integral description of the low-energy spectral edge.  The fluctuation
theory and physical significance of the spectral edge were further studied
in \cite{AltlandEtAl24}.  Variational and algorithmic approaches to the
ground-state problem for SYK and related interacting fermionic Hamiltonians
appear in \cite{HTS21,HO22,HSHT23}. Recently, Basu--Kothari--Midha
\cite{BKM26} proved
\[
\big(1-O(q^{-1/2}+q^4N^{-2})\big)
\frac{\sqrt{2N}}q
\leq\E\lambda_1
\leq\frac{\sqrt{2N}}q+O(1)
\]
for $4\leq q<\sqrt N/4$.  In particular, their bounds yield
$\lim_{N\to\infty}q\E\lambda_1/\sqrt N=\sqrt2$ when
$1\ll q\ll\sqrt N$. For fixed $q\geq4$, their result gives the upper bound
$\E\lambda_1\leq\sqrt{2N}/q+O_q(1)$, whereas the displayed lower estimate is not guaranteed to be positive. These results did not determine
the leading asymptotic location of $\lambda_1$ for fixed $q\geq4$ or $q=\Omega(\sqrt{N})$.

As Theorem \ref{mainthm2}\,(i) is valid for any fixed $\beta>0$, it has a spectral consequence after taking $\beta\to \infty$.  Recall $p_{{\rm SD},q}$ from \eqref{pSD}, and define
\begin{equation}
	e_{{\rm SD},q}
	\deq\lim_{\beta\to \infty} \frac{p_{{\rm SD},q}(\beta)}{\beta}\,,
	\quad \mbox{and} \quad 
	\kappa_{{\rm SD},q}
	\deq\sqrt{q2^q}\,e_{{\rm SD},q}\,.
	\label{1.10}
\end{equation}
We may now state our second main result.

\begin{theorem}[Spectral edge]
	\label{mainthm1}
(i)	Fix even $q\geq4$.  Then
	\begin{equation} \label{jjj}
		\lim_{N\to\infty}\frac{\lambda_1}{\sqrt N}
		=\kappa_{{\rm SD},q}
		\quad\mbox{with probability one}\,.
	\end{equation}
	The resulting fixed-$q$ constants satisfy
	\begin{equation} \label{kkk}
		\sqrt{\frac{2q}{q+1}}
		\leq q\kappa_{{\rm SD},q}\leq\sqrt{2-\frac1q}\,,\quad
		\lim_{\substack{q\to\infty}}
		q\kappa_{{\rm SD},q}=\sqrt2\,.
	\end{equation}

(ii) If $q\to \infty$ and $q\leq N/2$, then
\begin{equation} \label{1.11}
\lim_{N\to \infty}\frac{\lambda_1 \sqrt{1-\e^{-2q^2/N}}}{2}=1\,,
\end{equation}
with probability one. 
\end{theorem}

Theorem \ref{mainthm1}, together with \cite{BKM26}, determines the asymptotic edge location of the SYK model for all even interaction orders. It in particular answers the spectral-edge question posed by
Feng--Tian--Wei \cite[Section 1.4, item 1]{FengTianWei1}, where
they explicitly noted that for fixed $q\geq 4$, establishing the convergence of $\E\lambda_1/\sqrt N$ appeared very difficult.

For the upper bound when $q$ grows, our basic approach is the Gaussian
creation and annihilation operator method of Basu--Kothari--Midha
\cite{BKM26}. They use the Johnson scheme to compute the spectrum of the
matrix $((-1)^{|A\cap B|})_{A,B}$, indexed by $q$-subsets of $\qq N$.
Their comparison gives an error which grows with the degree in the Gaussian
variables and is too large when $q\asymp\sqrt N$. Our main observation is that
the polynomial matrix functions used in the comparison inherit the orthogonal
symmetry of the Majorana matrices. This expresses the quadratic form of the
Johnson adjacency matrix as a sum of squared commutators; see Lemma
\ref{lem4.4}. The resulting bound controls the main positive nonconstant
contribution independently of the polynomial degree. Together with estimates
of the remaining eigenvalues, this gives the sharp upper bound for all growing
$q\leq N/2$.
Thus we extend the subcritical result of \cite{BKM26} to the critical and
supercritical regimes $q=\Omega(\sqrt N)$. The lower bound in these latter
regimes follows from the global law of Feng--Tian--Wei
\cite[Theorem 1]{FengTianWei1}.

\begin{remark}
	(i) For $q=4$, the high-precision solution of the large-$N$
	Schwinger--Dyson equation in the normalization of \cite{ACT23} gives
	\[
	e_0=-0.04063026975834491522143475022673(6)\,.
	\]
	Gaussian symmetry identifies the upper and lower edges.
	In the thermodynamic normalization \eqref{1.1},
	$e_{{\rm SD},4}=|e_0|$.  Since \eqref{1.5} reads
	$\mathcal H=8H/\sqrt{aN}$ when $q=4$, the definition in
	\eqref{1.10} gives
	\begin{equation*}
		\kappa_{{\rm SD},4}=8|e_0|
		\approx0.32504215806675932177147800181384(48)\,.
	\end{equation*}
	
	(ii)  The relation \eqref{1.11} undergoes a phase transition at $q\asymp \sqrt{N}$. When $1\ll q\ll \sqrt{N}$, it implies
	\begin{equation} \label{1.12}
		\lim_{N\to \infty}\frac{q\lambda_1}{\sqrt{N}}=\sqrt{2}\,,
	\end{equation}
	which recovers the result of \cite{BKM26}. Moreover, \eqref{kkk} implies that \eqref{1.12} is consistent with \eqref{jjj}. When $q^2/N\to\alpha_*\in(0,\infty)$, write $\fs\deq\e^{-2\alpha_*}\in(0,1)$. We have 
	\[
	\lim_{N\to \infty}\lambda_1 =\frac{2}{\sqrt{1-\fs}}\,,
	\]
	and thus the edge is described by the $q$-Hermite law \eqref{qHermite}. Finally for $\sqrt{N}\ll q\leq N/2$, we have the semicircle behavior of the edge $\lim_{N\to \infty}\lambda_1=2$. Together with \cite{BKM26,FengTianWei1}, we now understand the asymptotic location of $\lambda_1$ for all even interaction orders of the SYK model.
\end{remark}

\begin{figure}[H] 
\centering
\includegraphics[width=.98\textwidth]{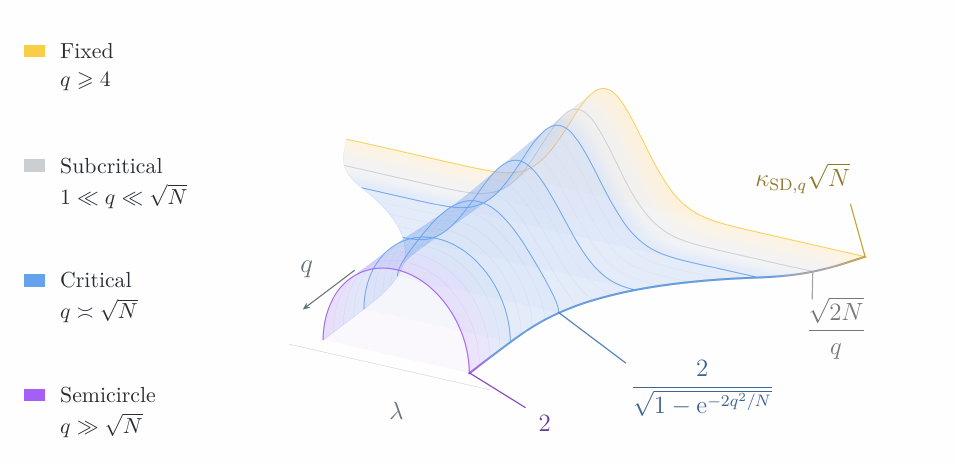}
\caption{Asymptotic spectrum and edges of the SYK model. The colors are inspired by the weapon rarity colors in Elden Ring Nightreign.}
\label{fig:edge-regimes}
\end{figure}

Theorem \ref{mainthm2} is pointwise in $\beta$. Once the spectral edge is known, we can also let the inverse temperature grow with $N$.

\begin{corollary}[Growing inverse temperature]
	\label{mainthm3} 
	Let $\beta_N\to\infty$ be arbitrary.
	
	 (i) Fix an even integer $q\geq4$. We have
	\begin{equation*} 
	\lim_{N\to \infty} f_{\mathrm{que}}(\beta_N)=-\lim_{N\to \infty}\frac{p_{\rm que}(\beta_N)}{\beta_N}
	= -e_{{\rm SD},q}  	\quad \mbox{and} \quad  \frac1{N\beta_N}\log\trn\e^{\beta_N H}
		\longrightarrow e_{{\rm SD},q}\,,
	\end{equation*}
with probability one for the second limit.

 (ii) Recall the definition of $b$ from \eqref{bbb}. Suppose $1\ll q\ll\sqrt N$. Then
\[
 \lim_{N\to\infty}\frac{q^2p_{\rm que}(b\beta_N)}{\beta_N}=1\,.
\]

(iii) Recall $d$ from \eqref{ddd}. Suppose $q\leq N/2$ and $q^2/N\to\alpha_*\in(0,\infty]$.
Set $\fs\deq\e^{-2\alpha_*}$. Then
\[
 \lim_{N\to\infty}\frac{Np_{\rm que}(\beta_N/d)}{\beta_N}
 =\frac2{\sqrt{1-\fs}}\,.
\]

\end{corollary}

\begin{remark}
	Theorems \ref{mainthm2}, \ref{mainthm1}, and Corollary \ref{mainthm3} treat the case $4\leq q\leq N/2$. By a simple symmetry argument, the analogues can also be obtained for interaction order $N-q$; see Section \ref{reflection} below. Thus our results are valid for all $q\in \{4,6,...,N-4\}$.
\end{remark}

\subsection*{Conventions} We write $\mathbb N\deq\{0,1,2,\ldots\}$ and $\mathbb N_+\deq\{1,2,\ldots\}$. We use $N$ as our fundamental large parameter. Unless stated otherwise, quantities implicitly depend on $N$, and we omit this dependence from our notation. Quantities independent of $N$ are called fixed or a constant. We use $c$ for a generic small positive constant and $C$ for a generic large positive constant; both may change from line to line. We use the usual big $O$ notation $O(\cdot)$, and if the implicit constant depends on a parameter $\alpha$ we indicate it by writing $O_\alpha(\cdot)$. For integers $k\geq r\geq0$, define
\[
(k)_r\deq k(k-1)\cdots(k-r+1)=\frac{k!}{(k-r)!}\,,
\quad (k)_0\deq1\,.
\]

\subsection*{Acknowledgment} The author is partially supported by National Key R\&D Program of China No.\,2023YFA1010400 and NSFC No.\,12322121.

\subsection*{Use of LLM} In an earlier version of this manuscript, GPT-5.6 assisted with the development of technical arguments on operator algebra. These steps have been removed from the current version. GPT-5.6 also assisted with literature search and expository writing, and GPT-6.1 drew Figure \ref{fig:edge-regimes}. The author is responsible for the content.

\section{The free energy} \label{sec2}
In this section we prove Theorem \ref{mainthm2}. Let us first recall the explicit forms of the Majorana fermions. Set $\psi_1^{(0)}\deq1$. The construction is inductive in dimension. For each $m\geq 1$, define
\begin{equation*}
	\psi_i^{(m)}\deq
	\begin{pmatrix}
		0&-\ii\psi_i^{(m-1)}\\
		\ii\psi_i^{(m-1)}&0
	\end{pmatrix}\quad \mbox{for}\  1\leq i\leq2m-1\,,
\end{equation*}
and
\[
	\psi_{2m}^{(m)}\deq
	\begin{pmatrix}
		0&I_{2^{m-1}}\\
		I_{2^{m-1}}&0
	\end{pmatrix}\,,\quad
	\psi_{2m+1}^{(m)}\deq
	\begin{pmatrix}
		I_{2^{m-1}}&0\\
		0&-I_{2^{m-1}}
	\end{pmatrix}\,.
\]
Here $\psi_i^{(m)}$ are matrices of dimension $2^m$. Finally, for $i=1,2,...,N$, set $\psi_i\deq\psi_i^{(N/2)}\in \bb C^{L\times L}$, and $\Gamma\deq\psi_{N+1}^{(N/2)}$. It is easy to check that
\begin{equation} \label{2.2}
	\psi_i\psi_j+\psi_j\psi_i=2\delta_{ij}\,.
\end{equation}
For each $A=\{i_1,...,i_q\}\subset \qq{N}$, $i_1<i_2<\cdots<i_q$, we set
\[
	\Psi_A\deq\ii^{q/2}\psi_{i_1}\cdots \psi_{i_q}\,,
\]
which is Hermitian, and the relation \eqref{2.2} guarantees \eqref{1.2}. For the rest of the paper, write
\begin{equation} \label{kappa}
	\varkappa\deq \bigg(\frac{(q-1)!}{2^qN^{q-1}}\bigg)^{1/2}\,, \quad \mbox{and thus} \quad H=\varkappa \sum_{A} J_A\Psi_A\,.
\end{equation}

Next are some elementary concentration results. In particular, we have a comparison between the annealed and quenched pressures.

\begin{lemma} \label{lemma2.1}
Let $4\leq q\leq N/2$ be an even integer. Let $W\in \bb C^{L\times L}$ satisfy $W\geq0$, $\tr W=1$, and set $\phi(D)\equiv \phi_W(D)\deq\tr(W D)$. We have
\begin{equation} \label{2.3}
	\binom Nq^{-1}\sum_{|A|=q}\phi(\Psi_A)^2
	\leq\frac1{N-q+1}\,.
\end{equation}
In addition,
\begin{equation} \label{2.5}
	0\leq p_{\rm ann}(\beta)-p_{\rm que}(\beta)
	\leq\frac{\beta^2a}{q2^{q+1}(N-q+1)}
\end{equation}
for all $\beta\geq0$. Here $a$ is defined as in \eqref{1.5}. Furthermore, for $t\geq0$,
\begin{equation} \label{2.4}
\begin{aligned}
	\bb P\big(|\lambda_1-\bb E\lambda_1|\geq t\big)
	\leq2\exp\bigg(-\frac{(N-q+1)t^2}{2}\bigg)\,,\quad \bb E|\lambda_1-\bb E\lambda_1|\leq(N-q+1)^{-1/2}\,.
\end{aligned}
\end{equation}
\end{lemma}
\begin{proof}
Consider a family of anticommuting Hermitian matrices $U_i$, with $U_i^2=I$. Then we have $(\sum_j\phi(U_j)U_j)^2=(\sum_j\phi(U_j)^2)I$. Set $S\deq\sum_j\phi(U_j)^2$, and Cauchy--Schwarz gives
\[
	S^2=\phi\Big(\sum_j\phi(U_j)U_j\Big)^2
	\leq\phi\Big(\Big(\sum_j\phi(U_j)U_j\Big)^2\Big)=S\,.
\]
Hence $S=\sum_j\phi(U_j)^2\leq1$. For each $T\subset \qq{N}$, $|T|=q-1$, the matrices $\Psi_{T\cup\{x\}}$, $x\notin T$, anticommute. As a result, we have
\[
	q\sum_{|A|=q}\phi(\Psi_A)^2
	=\sum_{|T|=q-1}\sum_{x\notin T}\phi(\Psi_{T\cup\{x\}})^2
	\leq\sum_{|T|=q-1}1=\frac{q}{N-q+1}\binom Nq\,,
\]
which proves \eqref{2.3}. Now, think ${\bf J}\deq(J_A)_{|A|=q}$ as a standard Gaussian vector in $\bb R^{\binom Nq}$, with norm $\|{\bf J}\|_2^2\deq\sum_{|A|=q}J_A^2$. The lower bound of \eqref{2.5} is merely Jensen's inequality. For the upper bound, take $W\deq\e^{\beta H}/\tr\e^{\beta H}$, and note that
\[
	\partial_{J_A}\log\trn\e^{\beta H}=\beta\varkappa\frac{\tr(\Psi_A \e^{\beta H})}{\tr \e^{\beta H}}= \beta\varkappa \phi(\Psi_A)\,.
\]
Applying \eqref{2.3} yields
\[
	\|\nabla_{\bf J} \log\trn\e^{\beta H(J)}\|_2^2=\beta^2\varkappa^2\sum_{|A|=q}\phi(\Psi_A)^2\leq \beta^2\varkappa^2 \binom Nq \frac{1}{N-q+1}=\frac{\beta^2aN}{q2^{q}(N-q+1)}\,.
\] 
The upper bound of \eqref{2.5} then follows from the Gaussian concentration inequality
\[
 0\leq\log\bb E\trn\e^{\beta H}-\bb E\log\trn\e^{\beta H}
 \leq\frac12\sup_{\bf J}\|\nabla_{\b J}\log\trn\e^{\beta H}\|_2^2\,.
\]

To prove \eqref{2.4}, note that $\lambda_1=\max_{W\geq0,\tr W=1}\tr(W\mathcal H)$. By Cauchy--Schwarz and \eqref{2.3}, we have for any ${\bf J},{\bf J'}$ that
\[
\begin{aligned}
 |\lambda_1({\bf J})-\lambda_1({\bf J'})|
 &\leq\sup_{W\geq0,\tr W=1}\bigg|\binom Nq^{-1/2}\sum_{|A|=q}(J_A-J_A')\tr(W\Psi_A)\bigg|\\
 &\leq\|{\bf J}-{\bf J'}\|_2\sup_{W\geq0,\tr W=1}\bigg(\binom Nq^{-1}\sum_{|A|=q}\tr(W\Psi_A)^2\bigg)^{1/2}
 \leq\frac{\|{\bf J}-{\bf J'}\|_2}{\sqrt{N-q+1}}\,.
\end{aligned}
\]
Thus ${\bf J}\mapsto\lambda_1({\bf J})$ is Lipschitz with constant $(N-q+1)^{-1/2}$. By Gaussian concentration, for every $s\in\bb R$, we have
\begin{equation} \label{123}
 \bb E\e^{s(\lambda_1-\bb E\lambda_1)}\leq\exp\bigg(\frac{s^2}{2(N-q+1)}\bigg)\,.
\end{equation}
Markov's inequality then gives the first estimate in \eqref{2.4}. For the second estimate, Cauchy--Schwarz and $\cosh(sx)\geq1+s^2x^2/2$ give
\[
 (\bb E|\lambda_1-\bb E\lambda_1|)^2
 \leq\frac2{s^2}\big(\bb E\cosh(s(\lambda_1-\bb E\lambda_1))-1\big)
 \leq\frac2{s^2}\bigg(\exp\Big(\frac{s^2}{2(N-q+1)}\Big)-1\bigg)\,.
\]
Taking $s\to0$ finishes the proof.
\end{proof}

\subsection{Auxiliary Majorana matrices} \label{sec2comparison}
For the rest of Sections \ref{sec2comparison} and \ref{sec2.4}, fix an even integer $q\geq4$. Let $r\in \bb N_+$. For any subset $T=\{i_1,...,i_r\}\subset\qq{N}$ with $i_1<\cdots<i_r$, set
\[
	\Psi_T\deq\ii^{r(r-1)/2}\psi_{i_1}\cdots\psi_{i_r}\,.
\]
Note that $\Psi_T^*=\Psi_T$, $\Psi_T^2=I$. For $|T|=q$, as $q$ is even and $\ii^{q(q-1)/2}=\ii^{q/2}$, this agrees with $\Psi$ in \eqref{1.1}. Recall that $\Gamma=\psi_{N+1}^{(N/2)}$. Fix $n\in\bb N_+$. For $1\leq j\leq n$ and $1\leq i\leq N$, set
\begin{equation} \label{etaxi}
	\eta_{ij}\deq\Gamma\otimes\psi_{2N(j-1)+2i-1}^{(nN)}\,,\quad
	\xi_{ij}\deq\Gamma\otimes\psi_{2N(j-1)+2i}^{(nN)}\,,
	\quad\psi_i\deq\psi_i^{(N/2)}\otimes I_{2^{nN}}\,.
\end{equation}
Thus $\psi_i,\eta_{ij},\xi_{ij}$ are Majorana fermions in dimension
\[
\mathcal L\deq 2^{N/2+nN}\,.
\]
As an abuse of notation, for every matrix $D\in \bb C^{L\times L}$, we also write $D$ for $D\otimes I_{2^{nN}}$, and the normalized trace now satisfies $\trn=\mathcal L^{-1}\tr$. Let $\b \gamma\in \bb R^{n}$ with $\gamma_j\geq0$ for $j\in \qq{n}$. Consider the Hamiltonian
\begin{equation} \label{2.7}
	\bb B\deq\beta H+\frac\ii2\sum_{i,j}(j-1)\eta_{ij}\xi_{ij}
	+\frac\ii2\sum_{i,j}\gamma_j\psi_i\eta_{ij}\,,
\end{equation}
and
\begin{equation} \label{Phi}
 Z\equiv Z(\beta,\b\gamma)\deq\bb E\trn\e^{\bb B}\,,\quad
\Phi(D)\deq\frac{\bb E\trn(\e^{\bb B}D)}Z\,.
\end{equation}
Define the enlarged pressure as
\[
	p(\beta,\b\gamma)\deq\frac1N\log\frac{Z(\beta,\b\gamma)}{Z(0,\b\gamma)}\,.
\]
Since $H$ and $\sum_{i,j}(j-1)\eta_{ij}\xi_{ij}$ act on separate tensor factors, we have
\begin{equation} \label{Zpp}
	Z(\beta,\b0)=Z(0,\b0)\bb E\trn\e^{\beta H}\,,\quad
	p(\beta,\b0)=p_{\rm ann}(\beta)\,,\quad \mbox{and} \quad  p(0,\b\gamma)=0\,.
\end{equation}
We have the following comparison, whose proof will be given in Section \ref{sec2.2} below.
\begin{lemma} \label{lem2.2}
For every $\beta\geq0$ and $\gamma_j\geq0$, $1\leq j\leq n$, we have
\begin{equation} \label{gamma^2}
	|p(\beta,\b\gamma)-p_{\rm ann}(\beta)|
	\leq\frac18\sum_{j=1}^n\gamma_j^2\,.
\end{equation}
\end{lemma}

At $\b\gamma=\b0$, the auxiliary matrices decouple from $H$ and their contribution cancels in \eqref{Zpp}. For nonzero $\b\gamma$, Lemma \ref{lem2.2} and $p(0,\b\gamma)=0$ give
\[
 \bigg|p_{\rm ann}(\beta)-\int_0^\beta\partial_u p(u,\b\gamma)\,\dd u\bigg|
 \leq\frac18\sum_{j=1}^n\gamma_j^2\,.
\]
Thus we shall estimate $\partial_\beta p(\beta,\b\gamma)$ using the auxiliary matrices and integrate in $\beta$. We then recover the pressure of $H$ by letting $\sum_j\gamma_j^2$ tend to zero after taking $N\to\infty$. For any matrix $D$, write 
$$
D(t)\deq \e^{-t\bb B}D\e^{t\bb B}\,.
$$
For $r\in \qq{N}$ and $t\in [0,1]$, define
\begin{equation} \label{Gr}
	\mathcal G_r(t)\deq \frac{1}{2^r\binom Nr}\sum_{|T|=r}\Phi(\Psi_T(t)\Psi_T)\,.
\end{equation}
In particular, we have
\[
	\mathcal G_1(t)=\frac{1}{2NZ}\sum_{i=1}^N \bb E\trn(\e^{(1-t)\bb B}\psi_i\e^{t\bb B}\psi_i)\quad \mbox{and} \quad \mathcal G_q(t)=\frac{1}{2^q\binom Nq}\sum_{|A|=q}\Phi(\Psi_A(t)\Psi_A)\,.
\]
As $\partial_\beta \bb B=H$, we get
\begin{equation} \label{2.9}
	\begin{aligned}
		\partial_\beta p(\beta,\b\gamma)
		&=\frac{1}{NZ}\bb E\trn(\e^{\bb B}H)
		=\frac{\varkappa}{NZ}\sum_{|A|=q}\bb E\big(J_A\trn(\e^{\bb B}\Psi_A)\big)\\
		&=\frac{\beta\varkappa^2}{NZ}\sum_{|A|=q}\int_0^1 \bb E\trn(\e^{(1-t)\bb B}\Psi_A\e^{t\bb B}\Psi_A)\,\dd t=\frac{\beta\varkappa^2}{N}\sum_{|A|=q}\int_0^1 \Phi(\Psi_A(t)\Psi_A)\,\dd t\,,
	\end{aligned}
\end{equation}
where in the third step we used Stein's Lemma, $\partial_{J_A}\bb B=\beta\varkappa\Psi_A$, and Duhamel's formula. Inserting \eqref{1.5} and \eqref{Gr} into \eqref{2.9} yields an important identity
\begin{equation} \label{2.9a}
	\partial_\beta p(\beta,\b\gamma)=\frac{\beta a}{q}\int_0^1 \mathcal G_q(t)\,\dd t\,.
\end{equation}
Thus our main goal is to estimate $\mathcal G_q$.

Fix $\ee\in (0,1]$, and abbreviate
\begin{equation} \label{2.10}
	\delta\deq \frac{\sqrt{\ee}}{q}\,,\quad \langle f\rangle\deq \delta^{-n}\int_{[\delta,2\delta]^n} f(\b\gamma)\,\dd \b\gamma\,.
\end{equation}
The following is our main technical step for understanding $\mathcal G_q$. The proof will be given in Section \ref{seccomparisonproof}.

\begin{proposition} \label{prop2.4}
	Fix an even integer $q\geq4$. For every fixed $\beta_*>0$, we have
	\begin{equation} \label{2.13}
		\lim_{n\to\infty}\limsup_{\ee\downarrow 0}\limsup_{N\to\infty}\sup_{0\leq\beta\leq\beta_*}\big\langle\|2^q\mathcal G_q-(2\mathcal G_1)^q\|_\infty\big\rangle=0\,.
	\end{equation}
\end{proposition}

\subsection{The Schwinger--Dyson equation} \label{sec2.4}
Now we need to compute $\cal G_1$. For $x\geq 0$ and $\tau\in [0,1]$, define
\[
	k_x(\tau)\deq \frac{\cosh(x(1/2-\tau))}{2\cosh(x/2)}=\frac{\e^{x\tau}+\e^{x(1-\tau)}}{2(1+\e^x)}\,.
\]
For a Borel probability measure $\mu$ on $[0,\infty)$, we define the normalized positive Lehmann kernel as
\[
	G(\tau)\equiv G_\mu(\tau)\deq \int_{[0,\infty)} k_x(\tau)\,\mu(\dd x)\,.
\]
It is easy to see that $0<G\leq 1/2$, $G(1-\tau)=G(\tau)$, and $G(0)=G(1)=1/2$. Let $U$ be a Hermitian matrix with $U^2=I$. Let $(x_j)$, $j=1,2,...,\mathcal L$, denote the eigenvalues of $\bb B$, with corresponding normalized eigenvectors $(\b v_j)$. Set
\[
	\nu\deq\frac{1}{\mathcal L Z}\sum_{j,k=1}^{\mathcal L}\bb E\,\e^{x_j}|\langle \b v_j,U\b v_k\rangle|^2\delta_{x_k-x_j}\,.
\]
The spectral decomposition gives
\[
	\Phi(U(t)U)=\frac{1}{\mathcal L Z}\sum_{j,k=1}^{\mathcal L}\bb E\,\e^{(1-t)x_j+tx_k}|\langle \b v_j,U\b v_k\rangle|^2=\int_{\bb R}\e^{tx}\,\nu(\dd x)\,.
\]
For each Borel set $S\subset [0,\infty)$, define
\[
	\nu_+(S)\deq \nu\big(\{x\in \bb R:|x|\in S\}\big)\,.
\]
Since $U^2=I$, $\nu_+$ is a probability measure. Interchanging $j$ and $k$ gives $\nu(-\dd x)=\e^x\nu(\dd x)$. Hence
\begin{equation} \label{2.21}
	\Phi(U(t)U)=\int_{[0,\infty)}\frac{\e^{tx}+\e^{(1-t)x}}{1+\e^x}\,\nu_+(\dd x)=2\int_{[0,\infty)} k_x(t)\,\nu_+(\dd x)\,.
\end{equation}
In particular, by \eqref{Gr}, $2^{r-1}\mathcal G_r$ is a normalized positive Lehmann kernel, so $0<2^r\mathcal G_r\leq 1$, $\mathcal G_r(1-t)=\mathcal G_r(t)$, and $2^r\mathcal G_r(0)=2^r\mathcal G_r(1)=1$.

We shall use the Fourier convention
\[
	\omega_k\deq (2k+1)\pi\,,\quad \wh f(k)\deq \int_0^1 \e^{\ii\omega_k t}f(t)\,\dd t\,,\quad k\in \bb Z\,.
\]
In convolutions, functions are extended antiperiodically by $f(t+m)=(-1)^m f(t)$ for $0<t<1$ and $m\in \bb Z$, and $(f*g)(t)\deq \int_0^1 f(t-s)g(s)\,\dd s$. It is easy to check that $\wh{f*g}(k)=\wh f(k)\wh g(k)$. The next result compares $\mathcal G_1$ with the Schwinger--Dyson kernel. Its proof will be given in Section \ref{sec2identification}.

\begin{lemma} \label{lem2.4}
	For each even integer $q\geq 4$ and $\beta\geq 0$, there is a unique normalized positive Lehmann kernel $G_\beta\equiv G_{\beta,q}$ satisfying
	\begin{equation} \label{1.3}
		\wh G_\beta(k)=-\big(\ii\omega_k+\beta^2\wh{G_\beta^{q-1}}(k)\big)^{-1}\,,\quad k\in \bb Z\,.
	\end{equation}
	The derivative of $G_\beta(t)$ satisfies
	\begin{equation} \label{2.12}
		G_\beta'=\beta^2(G_\beta^{q-1}*G_\beta)\,.
	\end{equation}
For fixed $q$ and every fixed $\beta_*>0$, we have
	\begin{equation} \label{2.14a}
		\lim_{n\to\infty}\limsup_{\ee\downarrow 0}\limsup_{N\to\infty}\sup_{0\leq\beta\leq\beta_*}\big\langle \|\mathcal G_1-G_\beta\|_\infty\big\rangle=0\,.
	\end{equation}
In addition, the map $\beta\mapsto G_\beta$ is continuous in the uniform norm.
\end{lemma}

We are ready to prove the first part of Theorem \ref{mainthm2}.

\begin{proof}[Proof of Theorem \ref{mainthm2}\,(i)]
Fix even $q\geq4$ and $\beta_*>0$. Here $\mathcal G_r$ is evaluated at inverse temperature $\beta$. Since $p(0,\b\gamma)=p_{{\rm SD},q}(0)=0$, and the averaging domain in \eqref{2.10} does not depend on $\beta$, we have
\[
\begin{aligned}
 \langle p(\beta,\b\gamma)\rangle-p_{{\rm SD},q}(\beta)
 =\int_0^\beta\big(\langle\partial_u p(u,\b\gamma)\rangle-p_{{\rm SD},q}'(u)\big)\,\dd u\,,\quad
 |p_{\rm ann}(\beta)-\langle p(\beta,\b\gamma)\rangle|
 \leq\frac18\sum_{j=1}^n\langle\gamma_j^2\rangle\leq\frac{n\ee}{2q^2}\,,
\end{aligned}
\]
where the second relation follows from Lemma \ref{lem2.2} and \eqref{2.10}. Thus
\[
 \sup_{0\leq\beta\leq\beta_*}|p_{\rm ann}(\beta)-p_{{\rm SD},q}(\beta)|
 \leq\frac{n\ee}{2q^2}+\int_0^{\beta_*}
 \big|\langle\partial_\beta p(\beta,\b\gamma)\rangle-p_{{\rm SD},q}'(\beta)\big|\,\dd\beta\,.
\]
By \eqref{pSD}, \eqref{2.9a}, $a\leq1$, and $0<2\mathcal G_1,2G_\beta\leq1$, we have
\[
\begin{aligned}
\int_0^{\beta_*} \big|\langle\partial_\beta p(\beta,\b\gamma)\rangle-p_{{\rm SD},q}'(\beta)\big|\,\dd\beta&=\int_0^{\beta_*}\frac\beta q\bigg|a\Big\langle\int_0^1\mathcal G_q(t)\,\dd t\Big\rangle-\int_0^1G_\beta(t)^q\,\dd t\bigg|\,\dd\beta\\
 &\leq\int_0^{\beta_*}\frac{\beta}{q2^q}\Big(1-a+
 \big\langle\|2^q\mathcal G_q-(2\mathcal G_1)^q\|_\infty\big\rangle
 +2q\big\langle\|\mathcal G_1-G_\beta\|_\infty\big\rangle\Big)\,\dd\beta\,.
\end{aligned}
\]
As $\lim_N a=1$, Proposition \ref{prop2.4} and Lemma \ref{lem2.4} imply that the RHS of the above tends to 0 when we take $N\to\infty$, then $\ee\downarrow0$, and finally $n\to\infty$. This proves $\lim_N p_{\rm ann}(\beta)=p_{{\rm SD},q}(\beta)$. The same limit for $p_{\rm que}(\beta)$ follows from Lemma \ref{lemma2.1}. This finishes the proof.
\end{proof}

\subsection{Growing \texorpdfstring{$q$}{q} below \texorpdfstring{$\sqrt N$}{sqrt N}} \label{secsub}
We now prove Theorem \ref{mainthm2}\,(ii). Recall $b=\sqrt{2^{q-1}/q}$ from \eqref{bbb}. The comparison with the Schwinger--Dyson pressure is given by the following result, whose proof uses the moment expansion and will be given in Section \ref{secgrowthpressure}.

\begin{proposition} \label{prop2.6}
Suppose $1\ll q\ll\sqrt N$. For every fixed $\beta>0$, we have
\[
 \lim_{N\to\infty}q^2\big(p_{\rm ann}(b\beta)-p_{{\rm SD},q}(b\beta)\big)=0\,.
\]
\end{proposition}

By \eqref{pSD} and Proposition \ref{prop2.6}, it remains to study the asymptotics of $(2G_{b\beta,q})^q$.

\begin{lemma} \label{lem2.4a}
Fix $\beta_*>0$ and let $0\leq \beta\leq \beta_*$. For even $q$, set $g(t)\deq q\log(2G_{b\beta,q}(t))$, $t\in [0,1]$.	Let $v\equiv v(\beta)\in [0,\pi/2)$ be the unique solution of $\beta\cos v=2v$, and define
	\begin{equation} \label{2.17}
	\mathfrak	g_\beta(t)\deq 2\log\frac{\cos v}{\cos(v(2t-1))}\,.
	\end{equation}
	Then for $q\geq \beta_*^2$, we have
	\begin{equation} \label{2.18}
		\|g- 	\mathfrak g_\beta\|_\infty\leq \frac{\beta^4}{q}\,.
	\end{equation}
\end{lemma}
\begin{proof}
Let us first prove
\begin{equation} \label{2.16a}
		\|g''-2\beta^2\e^g\|_\infty\leq \frac{8\beta^4}{q}\,.
	\end{equation}
Abbreviate $G\deq G_{b\beta,q}$. By \eqref{2.12} and $0<G\leq 1/2$, we have
\[
	|G'|\leq \frac{\beta^2}{2q}\quad \mbox{and} \quad G(t)\geq \frac{1}{2}-\frac{\beta^2}{2q}\min(t,1-t)\geq \frac{1}{2}-\frac{\beta^2}{4q}\geq \frac{1}{4}
\]
for $q\geq \beta_*^2$, as well as 
\[
G'=(b\beta)^2(G^{q-1}*G)=(b\beta)^2\int_0^t G(u)^{q-1}G(t-u)\,\dd u-(b\beta)^2\int_t^1 G(u)^{q-1}G(1+t-u)\,\dd u\,.
\]
Differentiating in $t$ and using $G(0)=G(1)=1/2$ gives $G''=(b\beta)^2 G^{q-1}+(b\beta)^2(G^{q-1}*G')$, and thus
\begin{equation} \label{2.19a}
\begin{aligned}
 g''=q\frac{G''}{G}-q\Big(\frac{G'}{G}\Big)^2=2\beta^2(2G)^{q-2}+\frac{q(b\beta)^2(G^{q-1}*G')}{G}-q\Big(\frac{G'}{G}\Big)^2\,.
\end{aligned}
\end{equation}
As a result,
\[
\begin{aligned}
 |g''-2\beta^2\e^g| &=|g''-2\beta^2(2G)^q|
\leq 2\beta^2\big((2G)^{q-2}-(2G)^q\big)
 +\frac{q(b\beta)^2|G^{q-1}*G'|}{G}+q\Big(\frac{G'}{G}\Big)^2\\
 &\leq 2\beta^2\frac{\beta^2}{q}
 +4q(b\beta)^2 2^{-(q-1)}\frac{\beta^2}{2q}
 +q\Big(\frac{2\beta^2}{q}\Big)^2=\frac{8\beta^4}{q}\,,
\end{aligned}
\]
where in the third step we used $0\leq (2G)^{q-2}-(2G)^q\leq 1-(2G)^2\leq \beta^2/q$, $|G^{q-1}*G'|\leq 2^{-(q-1)}\|G'\|_\infty$, $\|G'\|_\infty\leq \beta^2/(2q)$, $G\geq 1/4$. This proves \eqref{2.16a}.

Next we prove \eqref{2.18}. By \eqref{2.17}, we have $	\mathfrak g_\beta''=2\beta^2\e^{	\mathfrak g_\beta}$ and $	\mathfrak g_\beta(0)=	\mathfrak g_\beta(1)=0$. Note that $g=0$ for $\beta=0$. For $\beta>0$, set $\ee_\beta\deq \|g''-2\beta^2\e^g\|_\infty$. If $g-	\mathfrak g_\beta-\ee_\beta t(1-t)/2$ attains a positive maximum at some $t\in (0,1)$, then we have
\[
	0\geq \big(g-	\mathfrak g_\beta-\ee_\beta t(1-t)/2\big)''(t)\geq 2\beta^2\big(\e^{g(t)}-\e^{ 	\mathfrak g_\beta(t)}\big)>0\,,
\]
which is a contradiction. Thus we must have $g-	\mathfrak g_\beta\leq \ee_\beta t(1-t)/2\leq \ee_\beta/8$ for all $t\in[0,1]$. The matching lower bound is proved in the same way. Together with \eqref{2.16a}, we finish the proof.
\end{proof}

\begin{proof}[Proof of Theorem \ref{mainthm2}\,(ii)]
Suppose $1\ll q\ll\sqrt N$. Fix $\beta_*>0$ and let $0\leq\beta\leq\beta_*$. Let $g,\mathfrak g_\beta$ be as in Lemma \ref{lem2.4a}. As $g,\mathfrak g_\beta\leq0$, \eqref{2.18} gives
\[
 \|\e^g-\e^{\mathfrak g_\beta}\|_\infty\leq\|g-\mathfrak g_\beta\|_\infty\leq\frac{\beta^4}{q}\,.
\]
Note that the maximum in \eqref{bbb} is attained at the solution of $\beta\cos v=2v$. Thus $f(\beta)=\beta\sin v-v^2$, and
\[
 f'(\beta)=\sin v+(\beta\cos v-2v)\frac{\dd v}{\dd\beta}=\sin v\,,\quad f(0)=0\,.
\]
For $\beta>0$, by \eqref{2.17}, we have
\[
\begin{aligned}
 \frac{\beta}{2}\int_0^1\e^{\mathfrak g_\beta(t)}\,\dd t
 =\frac{\beta\cos^2v}{2}\int_0^1\frac{\dd t}{\cos^2(v(2t-1))}=\frac{\beta\cos^2v}{2v}\tan v=\sin v=f'(\beta)\,.
\end{aligned}
\]
This identity also holds at $\beta=0$ by continuity. Since $b^2=2^{q-1}/q$, combining the above with \eqref{pSD} yields
\[
 \bigg|q^2\frac{\dd}{\dd\beta}p_{{\rm SD},q}(b\beta)-f'(\beta)\bigg|
 =\frac{\beta}{2}\bigg|\int_0^1\big(\e^{g(t)}-\e^{\mathfrak g_\beta(t)}\big)\,\dd t\bigg|
 \leq\frac{\beta^5}{2q}\,,
\]
and thus
\begin{equation*} 
 \sup_{0\leq\beta\leq\beta_*}\big|q^2p_{{\rm SD},q}(b\beta)-f(\beta)\big|
 \leq\frac{\beta_*^6}{12q}\,.
\end{equation*}
Together with Proposition \ref{prop2.6}, this proves $\lim_Nq^2p_{\rm ann}(b\beta)=f(\beta)$ as desired. By Lemma \ref{lemma2.1},
\[
 0\leq q^2\big(p_{\rm ann}(b\beta)-p_{\rm que}(b\beta)\big)
 \leq\frac{a\beta^2}{4(N-q+1)}\to0\,,
\]
which proves the quenched limit. This finishes the proof.
\end{proof}

\section{Proof of the correlation estimates} \label{seccomparisonproof}
In this section, we prove Lemma \ref{lem2.2}, Proposition \ref{prop2.4} and Lemma \ref{lem2.4}. Fix an even integer $q\geq4$ and $\beta_*>0$. Throughout the section, $0\leq\beta\leq\beta_*$, and $\Phi$ is the annealed expectation for $\bb B$ in \eqref{2.7}. Recall $\mathcal L=2^{N/2+nN}$, $\trn=\mathcal L^{-1}\tr$, $\delta=\sqrt{\ee}/q$, and the average $\langle\cdot\rangle$ from \eqref{2.10}. 

\subsection{Identities for the auxiliary matrices} \label{sec2.2}
For $D_1,...,D_l\in \{\psi_i,\eta_{ij},\xi_{ij}:i\in\qq{N},\,j\in\qq{n}\}$ and $t_1,...,t_l\in [0,1]$, let $\pi$ be the permutation of $\qq{l}$ with $t_{\pi(1)}\geq\cdots\geq t_{\pi(l)}$. Among equal times, the listed order is preserved. The Fermionic time ordering is defined by
\[
 \mathcal T\big(D_1(t_1)\cdots D_l(t_l)\big)\deq
 \sgn(\pi)D_{\pi(1)}(t_{\pi(1)})\cdots D_{\pi(l)}(t_{\pi(l)})\,.
\]
If a factor is a product of matrices from $\{\psi_i,\eta_{ij},\xi_{ij}\}$ at the same time $t$, we view it as a product of individual factors at time $t$ inside $\mathcal T$. For $r\in\qq{q}$ and $\b t=(t_1,...,t_r)\in(0,1)^r$, we extend \eqref{Gr} by setting
\[
 \mathcal G_r(\b t)\deq\frac{1}{2^r\binom Nr}
 \sum_{i_1<\cdots<i_r}\Phi\bigg(\mathcal T\prod_{k=1}^r
 [\psi_{i_k}(t_k)\psi_{i_k}(0)]\bigg)\,,
\]
and we extend $\mathcal G_r$ to $[0,1]^r$ by continuity. Note that for $T=\{i_1<\cdots<i_r\}$,
\[
 \mathcal T\prod_{k=1}^r[\psi_{i_k}(t)\psi_{i_k}(0)]
 =(-1)^{r(r-1)/2}\prod_{k=1}^r\psi_{i_k}(t)\prod_{k=1}^r\psi_{i_k}(0)
 =\Psi_T(t)\Psi_T\,,
\]
and thus $\mathcal G_r(t,...,t)$ agrees with $\mathcal G_r(t)$ in \eqref{Gr}.

For a transposition $\sigma=(i\ k)$ of $\qq N$, consider the unitary matrix
\[
 \frac{\psi_i-\psi_k}{\sqrt2}\prod_{j=1}^n
 \frac{\eta_{ij}-\eta_{kj}}{\sqrt2}\frac{\xi_{ij}-\xi_{kj}}{\sqrt2}\,.
\]
By \eqref{2.2} and \eqref{etaxi}, conjugation by this product sends $\psi_l,\eta_{lj},\xi_{lj}$ to $-\psi_{\sigma(l)},-\eta_{\sigma(l)j},-\xi_{\sigma(l)j}$ respectively. The two auxiliary sums in \eqref{2.7} are unchanged, while $\Psi_A$ is sent to $\pm\Psi_{\sigma(A)}$. These signs do not affect the joint distributions of $(J_A)$, and the minus signs cancel in each pair $\psi_l(t)\psi_l(0)$. Since every permutation is a product of transpositions, for any distinct $i_1,...,i_r\in\qq N$, we have
\begin{equation} \label{2.19b}
 \Phi\bigg(\mathcal T\prod_{k=1}^r[\psi_{i_k}(t_k)\psi_{i_k}(0)]\bigg)=2^r\mathcal G_r(\b t)\,.
\end{equation}

For matrices $D_1,...,D_l$, H\"older's inequality for traces gives
\begin{equation} \label{2.19}
 \big|\trn(\e^{a_0\bb B}D_1\e^{a_1\bb B}\cdots D_l\e^{a_l\bb B})\big|
 \leq\trn\e^{\bb B}\prod_{k=1}^l\|D_k\|\,,\quad
 a_k\geq0\,,\quad\sum_k a_k=1\,,
\end{equation}
which implies $|2^r\mathcal G_r|\leq1$. To estimate integrals of $\mathcal G_q$, we shall compare them with moments of the following averages. For $1\leq j\leq n$, set
\begin{equation} \label{XY}
 X_j\deq\frac1N\sum_{i=1}^N\ii\psi_i\eta_{ij}\quad\mbox{and}\quad
 Y_j\deq\frac1N\sum_{i=1}^N\ii\psi_i\xi_{ij}\,,
\end{equation}
which are Hermitian matrices with norm at most $1$. Similar to $k_x(\tau)$, for $x\geq0$ and $0\leq\tau\leq1$, define
\[
 \ell_x(\tau)\deq\frac{\sinh(x(1/2-\tau))}{2\cosh(x/2)}\,.
\]

Fix $i\in\qq N$ and $j\in\qq n$, and abbreviate $\psi\deq\psi_i$, $\eta\deq\eta_{ij}$, and $\xi\deq\xi_{ij}$. Let $\mathcal U=D_1(t_1)\cdots D_l(t_l)$ be a product of an odd number of matrices with $D_a\in\{\psi_k,\eta_{km},\xi_{km}:k\in\qq{N},\,m\in\qq{n}\}$ and $t_a\in [0,1]$, where $\eta_{ij}$ and $\xi_{ij}$ are excluded from $\mathcal U$.

\begin{lemma} \label{lem2.7}
 We have
 \begin{equation} \label{2.22}
 \Phi(\mathcal T\eta(0+)\mathcal U)=\ii\gamma_j\int_0^1 k_{j-1}(t)\Phi(\mathcal T\psi(t)\mathcal U)\,\dd t\,,\quad
 \Phi(\mathcal T\xi(0+)\mathcal U)=\gamma_j\int_0^1\ell_{j-1}(t)\Phi(\mathcal T\psi(t)\mathcal U)\,\dd t\,.
 \end{equation}
 In particular,
 \begin{equation} \label{2.16}
 \Phi(X_j)=2\gamma_j\int_0^1k_{j-1}(t)\mathcal G_1(t)\,\dd t\,,\quad
 \Phi(Y_j)=0\,,\quad 0\leq\Phi(X_j)\leq\frac{\gamma_j}{2}\,.
 \end{equation}
\end{lemma}
\begin{proof}
Set
\[
 F_D(t)\deq\Phi(\mathcal T D(t)\mathcal U)\,,\quad D\in\{\psi,\eta,\xi\}\,.
\]
For $D\in\{\eta,\xi\}$, since $\eta$ and $\xi$ are excluded from $\mathcal U$, $D$ anticommutes with $D_a$ for every $a\in\qq{l}$. At a crossing with $D_a(t_a)$, time ordering gives
\[
\begin{aligned}
 \lim_{t\uparrow t_a}\mathcal T\big(D(t)D_a(t_a)\big)=-D_a(t_a)D(t_a)\,,\quad
 \lim_{t\downarrow t_a}\mathcal T\big(D(t)D_a(t_a)\big)=D(t_a)D_a(t_a)\,.
\end{aligned}
\]
These limits agree because
\[
 D(t_a)D_a(t_a)+D_a(t_a)D(t_a)
 =\e^{-t_a\bb B}(DD_a+D_aD)\e^{t_a\bb B}=0\,.
\]
If several factors have time $t_a$, the same cancellation applies to each of them. Thus $F_\eta$ and $F_\xi$ are continuous on $[0,1]$. Since $\mathcal U$ contains an odd number of factors, by cyclicity of the trace we have, for $D\in\{\eta,\xi\}$,
\begin{equation} \label{2.24a}
 F_D(1-)=F_D(1)=\Phi(D(1)\mathcal T\mathcal U)
 =\Phi((\mathcal T\mathcal U)D(0))=-F_D(0)=-F_D(0+)\,.
\end{equation}
Set $F_\pm\deq F_\eta\pm\ii F_\xi$. Since $D'(t)=-[\bb B,D](t)$, the commutators
\[
 [\bb B,\eta]=-\ii(j-1)\xi+\ii\gamma_j\psi\,,\quad
 [\bb B,\xi]=\ii(j-1)\eta\,,
\]
give $F_\pm'(t)=\pm(j-1)F_\pm(t)-\ii\gamma_jF_\psi(t)$ on each interval between the times in $\mathcal U$. Thus
\[
 \frac{\dd}{\dd t}\big(\e^{\mp(j-1)t}F_\pm(t)\big)
 =-\ii\gamma_j\e^{\mp(j-1)t}F_\psi(t)\,.
\]
We integrate on each interval between the times in $\mathcal U$ and sum over the intervals. The intermediate boundary terms cancel by the continuity of $F_\pm$, and \eqref{2.24a} gives
\begin{equation} \label{luoke}
 -(1+\e^{\mp(j-1)})F_\pm(0+)
 =-\ii\gamma_j\int_0^1\e^{\mp(j-1)t}F_\psi(t)\,\dd t\,.
\end{equation}
Note that
\[
 \frac{\e^{-xt}}{1+\e^{-x}}=k_x(t)+\ell_x(t)\,,\quad
 \frac{\e^{xt}}{1+\e^x}=k_x(t)-\ell_x(t)\,,
\]
and thus \eqref{luoke}  gives \eqref{2.22}.

To prove \eqref{2.16}, taking $\mathcal U=\psi_i(0)$ and applying \eqref{2.22} give
\[
 \Phi(\ii\psi_i\eta_{ij})=-\ii F_\eta(0+)
 =\gamma_j\int_0^1 k_{j-1}(t)\Phi(\psi_i(t)\psi_i)\,\dd t\,.
\]
Averaging over $i\in\qq{N}$ proves the first relation in \eqref{2.16}, and the last relation follows from $0<2\mathcal G_1\leq1$ and $\|k_{j-1}\|_1\leq1/2$. Similarly, as $\ell_{j-1}(1-t)\mathcal G_1(1-t)=-\ell_{j-1}(t)\mathcal G_1(t)$, we have
\[
 \Phi(Y_j)=-2\ii\gamma_j\int_0^1\ell_{j-1}(t)\mathcal G_1(t)\,\dd t=0\,,
\]
This finishes the proof.
\end{proof}

\begin{proof}[Proof of Lemma \ref{lem2.2}]
By \eqref{2.16}, we have
\[
 0\leq\partial_{\gamma_j}\big(N^{-1}\log Z(\beta,\b\gamma)\big)=\frac12\Phi(X_j)\leq\frac{\gamma_j}{4}\,.
\]
Integrating from $\b0$ to $\b\gamma$ gives
\[
 0\leq\frac1N\log\frac{Z(\beta,\b\gamma)}{Z(\beta,\b0)}\leq\frac18\sum_{j=1}^n\gamma_j^2\,.
\]
In particular, the same bound holds at $\beta=0$. Note that by \eqref{Zpp}, we have
\[
 p(\beta,\b\gamma)-p_{\rm ann}(\beta)=\frac1N\log\frac{Z(\beta,\b\gamma)}{Z(\beta,\b0)}-\frac1N\log\frac{Z(0,\b\gamma)}{Z(0,\b0)}\,,
\]
and thus we obtain \eqref{gamma^2} as desired. 
\end{proof}

\subsection{Variance estimates}
Recall $\Phi$ from \eqref{Phi}. For a Hermitian matrix $W$, define
\[
 \var_\Phi W\deq\Phi(W^2)-\Phi(W)^2\,,\quad
 \var_\Phi^{\rm D}W\deq\int_0^1\Phi(W(t)W)\,\dd t-\Phi(W)^2\,.
\]
Recall $X_j, Y_j$ from \eqref{XY}. We have the following result.

\begin{lemma} \label{lem2.6}
Fix $\beta_*>0$. Uniformly for $0\leq\beta\leq\beta_*$, we have
\begin{equation} \label{2.22a}
 \bigg\langle\sum_{j=1}^n(\var_\Phi X_j+\var_\Phi Y_j)\bigg\rangle
 \leq\frac{C_{\beta_*,q,n}}{N}\,.
\end{equation}
\end{lemma}
\begin{proof}
For $W_1,...,W_q\in\{X_1,Y_1,...,X_n,Y_n\}$, expanding the difference one factor at a time and applying Cauchy--Schwarz yield
\begin{equation} \label{2.26a}
 \Big|\Phi(W_1\cdots W_q)-\prod_{k=1}^q\Phi(W_k)\Big|
 \leq q\bigg(\sum_{j=1}^n(\var_\Phi X_j+\var_\Phi Y_j)\bigg)^{1/2}\,,
\end{equation}
where we used $\|W_k\|\leq1$ and $|\Phi(W_l)|\leq1$.

Since $\partial_{\gamma_j}\bb B=NX_j/2$, Duhamel's formula and \eqref{2.16} give
\[
 \partial_{\gamma_j}\Phi(X_j)=\frac N2\var_\Phi^{\rm D}X_j\,,\quad \mbox{and thus}\quad
 \frac1\delta\int_\delta^{2\delta}\var_\Phi^{\rm D}X_j\,\dd\gamma_j
 =\frac{2(\Phi(X_j)|_{2\delta}-\Phi(X_j)|_\delta)}{N\delta}\leq\frac2N\,.
\]
For $Y_j$, set $\bb B_t\deq\bb B+NtY_j/2$ for $t\in\bb R$, and define
\[
 \Phi_t(D)\deq\frac{\bb E\trn(\e^{\bb B_t}D)}{\bb E\trn\e^{\bb B_t}}\,.
\]
As one readily shows that,
\[
 \bigg[\frac12\sum_{i=1}^N\eta_{ij}\xi_{ij},\bb B_t\bigg]
 =\frac N2(tX_j-\gamma_jY_j)\,,
\]
by cyclicity, we have $\gamma_j\Phi_t(Y_j)=t\Phi_t(X_j)$. Differentiating this identity at $t=0$, and using Duhamel's formula and \eqref{2.16},  we get
\[
 \frac N2\var_\Phi^{\rm D}Y_j
 =\partial_t\Phi_t(Y_j)\big|_{t=0}
 =\frac{\Phi(X_j)}{\gamma_j}\leq\frac12\,,
\]
and thus $\var_\Phi^{\rm D}Y_j\leq N^{-1}$.

To apply \eqref{2.26a}, we need the ordinary variances. For every Hermitian matrix $W$, we shall show that
\begin{equation} \label{2.25}
 0\leq\var_\Phi W-\var_\Phi^{\rm D}W
 \leq\frac1{12}\Phi([W,[W,\bb B]])\,.
\end{equation}
Indeed, defining $\nu_+$ as in \eqref{2.21} with $W$ in place of $U$, we have $\nu_+([0,\infty))=\Phi(W^2)$, and the spectral decomposition gives
\[
 \Phi(W(t)W)=2\int_{[0,\infty)} k_x(t)\,\nu_+(\dd x)\quad\mbox{and}\quad
 \Phi([W,[W,\bb B]])=2\int_{[0,\infty)} x\tanh(x/2)\,\nu_+(\dd x)\,.
\]
Since $\sinh y\leq y\cosh y\leq(1+y^2/3)\sinh y$ for $y\geq0$, we have
\[
 0\leq k_x(0)-\int_0^1 k_x(t)\,\dd t=\frac12-\frac{\tanh(x/2)}{x}\leq\frac{x}{12}\tanh(x/2)\,.
\]
Integrating the above w.r.t.\,$2\nu_+$ proves \eqref{2.25}.

Now take $W=X_j$ or $Y_j$, and write $D_{ij}\deq\eta_{ij}$ or $\xi_{ij}$ accordingly. Since $\ii\psi_i D_{ij}$ commutes with $\Psi_A$ for $i\notin A$, and $[\ii\psi_iD_{ij},\ii\psi_kD_{kj}]=0$ for $i\neq k$, we have for each $q$-set $A$ that
\[
 [W,[W,\Psi_A]]
 =\frac1{N^2}\sum_{i,k\in A}[\ii\psi_iD_{ij},[\ii\psi_kD_{kj},\Psi_A]]\,,\quad \mbox{and thus} \quad 
 \|[W,[W,\Psi_A]]\|\leq\frac{4q^2}{N^2}\,.
\]
As in \eqref{2.9}, $\Phi([W,[W,\beta H]])$ can be computed via Stein's Lemma. Together with Duhamel's formula and \eqref{2.19}, we get
\begin{equation} \label{2.25a}
\begin{aligned}
 |\Phi([W,[W,\beta H]])|
 =\Big|\beta^2\varkappa^2\sum_{|A|=q}\int_0^1
 \Phi(\Psi_A(t)[W,[W,\Psi_A]])\,\dd t\Big|\leq\frac{4q^2\beta^2\varkappa^2}{N^2}\binom Nq\leq\frac{C_{\beta_*,q}}N\,.
\end{aligned}
\end{equation}
For the remaining terms in $\bb B$, a direct computation via \eqref{2.2} gives
\begin{equation} \label{2.25b}
 \Big[W,\Big[W,\frac\ii2\sum_{i,k}(k-1)\eta_{ik}\xi_{ik}\Big]\Big]
 =\frac{2\ii(j-1)}{N^2}\sum_{i=1}^N\eta_{ij}\xi_{ij}\,,
\end{equation}
as well as
\begin{equation} \label{2.25c}
 [X_j,[X_j,X_k]]=\frac4{N^2}\mathbf{1}_{j\neq k}X_k\quad\mbox{and}\quad
 [Y_j,[Y_j,X_k]]=\frac4{N^2}X_k\,.
\end{equation}
Since $\|X_k\|\leq1$ and $\gamma_k\leq2\delta\leq2$, we obtain from \eqref{2.7} and \eqref{2.25}--\eqref{2.25c} that
\[
 0\leq\Phi([W,[W,\bb B]])
 \leq\frac{C_{\beta_*,q}+2(j-1)+2\sum_{k=1}^n\gamma_k}{N}
 \leq\frac{C_{\beta_*,q,n}}N\,.
\]
Combining this with \eqref{2.25}, $\langle\var_\Phi^{\rm D}X_j\rangle\leq2/N$, and $\var_\Phi^{\rm D}Y_j\leq1/N$ yields \eqref{2.22a}. This finishes the proof.
\end{proof}

\subsection{Factorization of the correlations} \label{sec2.3}
We first prove a straightforward estimate that will be used in proving both Proposition \ref{prop2.4} and Lemma \ref{lem2.4}.

\begin{lemma} \label{lem2.7a}
For $r\in\qq{q}$, we have $|2^r\mathcal G_r|\leq 1$, and $\cal G_r$ is Lipschitz on $[0,1]^r$. On each region with a fixed ordering of $t_1,...,t_r$, we have
\begin{equation} \label{2.28}
 |\partial_{t_k}(2^r\mathcal G_r)(\b t)|
 \leq2\beta^2\varkappa^2\binom{N-1}{q-1}+\sum_{j=1}^n\gamma_j\,.
\end{equation}
\end{lemma}
\begin{proof}
Fix distinct $i_1,...,i_r\in\qq{N}$, and set
\[
 \mathcal U_k\deq\psi_{i_k}(0)\prod_{1\leq l\leq r,\,l\neq k}[\psi_{i_l}(t_l)\psi_{i_l}(0)]\,.
\]
When $t_1,...,t_r\in (0,1)$ are distinct, moving $\psi_{i_k}(t_k)$ to the front of $\mathcal U_k$ inside $\mathcal T$ produces no sign change. By \eqref{2.19b}, we have
\[
 \partial_{t_k}(2^r\mathcal G_r)(\b t)
 =-\Phi\big(\mathcal T[\bb B,\psi_{i_k}](t_k)\mathcal U_k\big)\,.
\]
From the definition of $\bb B$ in \eqref{2.7}, we have
\[
 -[\bb B,\psi_{i_k}]
 =-2\beta\varkappa\sum_{A\ni i_k}J_A\Psi_A\psi_{i_k}+\ii\sum_{j=1}^n\gamma_j\eta_{i_kj}\,.
\]
By Stein's Lemma and Duhamel's formula, for each $A$ containing $i_k$,
\[
 \Phi\big(\mathcal T J_A(\Psi_A\psi_{i_k})(t_k)\mathcal U_k\big)
 =\beta\varkappa\int_0^1
 \Phi\big(\mathcal T\Psi_A(v)(\Psi_A\psi_{i_k})(t_k)\mathcal U_k\big)\,\dd v\,,
\]
and thus
\begin{equation} \label{2.28a}
\begin{aligned}
 \partial_{t_k}(2^r\mathcal G_r)(\b t)
 =-2\beta^2\varkappa^2\sum_{A\ni i_k}\int_0^1
 \Phi\big(\mathcal T\Psi_A(v)(\Psi_A\psi_{i_k})(t_k)\mathcal U_k\big)\,\dd v+\ii\sum_{j=1}^n\gamma_j\Phi\big(\mathcal T\eta_{i_kj}(t_k)\mathcal U_k\big)\,.
\end{aligned}
\end{equation}
By \eqref{2.19}, each expectation on the RHS has modulus at most $1$. Since there are $\binom{N-1}{q-1}$ sets $A$ containing $i_k$, \eqref{2.28a} gives \eqref{2.28} when $t_1,...,t_r$ are distinct. For $k\neq l$, we have $\psi_{i_k}(t)\psi_{i_l}(t)=-\psi_{i_l}(t)\psi_{i_k}(t)$. As in the proof of Lemma \ref{lem2.7}, this identity cancels the sign change from $\mathcal T$, so $\mathcal G_r$ is continuous across $t_k=t_l$. Integrating \eqref{2.28} along the coordinates proves that $2^r\mathcal G_r$ is Lipschitz on $[0,1]^r$. This finishes the proof.
\end{proof}

\begin{proof}[Proof of Proposition \ref{prop2.4}]
Fix $W_k=X_{j_k}$ or $Y_{j_k}$ for $1\leq k\leq q$, and set $D_k(i)\deq\eta_{ij_k}$ or $\xi_{ij_k}$ and $L_k\deq k_{j_k-1}$ or $-\ii\ell_{j_k-1}$ accordingly. Expanding $W_1\cdots W_q$ yields
\[
 \Phi(W_1\cdots W_q)=\frac{1}{N^q}\sum_{i_1,...,i_q=1}^N \Phi\bigg(\prod_{k=1}^q [\ii\psi_{i_k}D_k(i_k)]\bigg)\,.
\]
For distinct indices $i_1,...,i_q\in\qq{N}$ and $0\leq l\leq q$, set
\[
 \mathcal U_l(\b t)\deq\mathcal T
 \prod_{k=1}^l[\psi_{i_k}(t_k)\psi_{i_k}(0)]
 \prod_{k=l+1}^q[\ii\psi_{i_k}(0)D_k(i_k)(0)]\,.
\]
For $1\leq l\leq q$, set
\[
 \mathcal U\deq\psi_{i_l}(0)
 \prod_{k=1}^{l-1}[\psi_{i_k}(t_k)\psi_{i_k}(0)]
 \prod_{k=l+1}^q[\ii\psi_{i_k}(0)D_k(i_k)(0)]\,.
\]
The product $\mathcal U$ has $2q-1$ Majorana factors. Since $i_1,...,i_q$ are distinct, it contains neither $\eta_{i_lj_l}$ nor $\xi_{i_lj_l}$. Using $\psi_{i_l}D_l(i_l)=-D_l(i_l)\psi_{i_l}$ and moving the two factors at site $i_l$ past the $2(l-1)$ preceding factors, we get
\[
 \Phi(\mathcal U_{l-1})=-\ii\Phi(\mathcal T D_l(i_l)(0+)\mathcal U)\,,\quad
 \Phi(\mathcal U_l)=\Phi(\mathcal T\psi_{i_l}(t_l)\mathcal U)\,.
\]
Thus \eqref{2.22} yields
\[
 \Phi(\mathcal U_{l-1})=\gamma_{j_l}\int_0^1
 L_l(t_l)\Phi(\mathcal U_l)\,\dd t_l\,,\quad 1\leq l\leq q\,.
\]
Applying this relation inductively for $l=1,...,q$ and using $\Phi(\mathcal U_q(\b t))=2^q\mathcal G_q(\b t)$ from \eqref{2.19b}, we get
\[
 \Phi\bigg(\prod_{k=1}^q [\ii\psi_{i_k}D_k(i_k)]\bigg)=\bigg(\prod_{k=1}^q\gamma_{j_k}\bigg)\int_{[0,1]^q}2^q\mathcal G_q(\b t)\prod_{k=1}^q L_k(t_k)\,\dd\b t\,.
\]
Since there are $(N)_q$ distinct tuples $(i_1,...,i_q)$ and $1-(N)_q/N^q\leq \binom q2/N$, we obtain
\begin{equation} \label{2.27}
 \bigg|\Phi(W_1\cdots W_q)
 -\frac{(N)_q}{N^q}\bigg(\prod_{k=1}^q\gamma_{j_k}\bigg)
 \int_{[0,1]^q}2^q\mathcal G_q(\b t)\prod_{k=1}^q L_k(t_k)\,\dd\b t\bigg|
 \leq\frac{\binom q2}{N}\,.
\end{equation}
By \eqref{2.19} and the preceding exact identity, removing the factor $(N)_q/N^q$ from \eqref{2.27} costs at most another $\binom q2/N$. On the other hand, \eqref{2.16} gives $\Phi(W_k)=\gamma_{j_k}\int_0^1L_k(t)2\mathcal G_1(t)\,\dd t$. Combining \eqref{2.26a}, \eqref{2.27}, Lemma \ref{lem2.6}, and $\gamma_{j_k}\geq\delta$, we arrive at
\begin{equation} \label{2.26}
\begin{aligned}
 \sup_{0\leq\beta\leq\beta_*}\bigg\langle\bigg|
 \int_{[0,1]^q}\bigg(2^q\mathcal G_q(\b t)-\prod_{k=1}^q 2\mathcal G_1(t_k)\bigg)
 \prod_{k=1}^q L_k(t_k)\,\dd\b t\bigg|\bigg\rangle\leq\delta^{-q}\bigg(\frac{C_{\beta_*,q,n}}{\sqrt N}+\frac{2\binom q2}{N}\bigg)\to0
\end{aligned}
\end{equation}
as $N \to \infty$.

It remains to deduce uniform convergence on $[0,1]^q$ from \eqref{2.26}. Set
\[
 f(\b t)\deq 2^q\mathcal G_q(\b t)-\prod_{k=1}^q 2\mathcal G_1(t_k)\,.
\]
By Lemma \ref{lem2.7a}, when $n\delta\leq1$, we have $\|f\|_\infty\leq2$ and $f$ is Lipschitz on $[0,1]^q$ with a constant $C_{\beta_*,q}$ depending only on $\beta_*$ and $q$. Note that
\[
 2\cosh\Big(\frac{j-1}{2}\Big)\big(k_{j-1}(t)+\ell_{j-1}(t)\big)
 =\e^{(j-1)(1/2-t)}\,.
\]
Thus every polynomial of degree at most $n-1$ in $\e^{1/2-t_1},...,\e^{1/2-t_q}$ is a linear combination of the products appearing in \eqref{2.26}.
Fix an arbitrary $h>0$, and choose a finite $h$-net $\b t_1,...,\b t_l$ of $[0,1]^q$. For each $1\leq a\leq l$, set $\phi_a(\b t)\deq c_a^{-1}(h-|\b t-\b t_a|)_+$, where $c_a>0$ is chosen so that $\int_{[0,1]^q}\phi_a(\b t)\,\dd\b t=1$. By the Stone--Weierstrass theorem, there is a polynomial $P_a$ in $\e^{1/2-t_1},...,\e^{1/2-t_q}$ such that $\|P_a-\phi_a\|_1\leq h$. Using $\|f\|_\infty\leq2$ and the Lipschitz bound for $f$, we have
\[
 \|f\|_\infty\leq2C_{\beta_*,q}h+\max_{1\leq a\leq l}\bigg|\int_{[0,1]^q}f\phi_a\,\dd\b t\bigg|
 \leq2(C_{\beta_*,q}+1)h+\max_{1\leq a\leq l}\bigg|\int_{[0,1]^q}fP_a\,\dd\b t\bigg|\,.
\]
Fix any $n$ larger than the degrees of $P_1,...,P_l$, and take $\ee$ small enough that $n\delta\leq1$. Each $\int fP_a$ is a finite linear combination of the integrals in \eqref{2.26}, with coefficients independent of $N,\beta,\b\gamma$. Since the family is finite, taking $N\to\infty$ and then $\ee\downarrow0$ in \eqref{2.26} gives
\[
 \limsup_{\ee\downarrow0}\limsup_{N\to\infty}
 \sup_{0\leq\beta\leq\beta_*}\langle\|f\|_\infty\rangle
 \leq2(C_{\beta_*,q}+1)h\,.
\]
Letting $n\to\infty$ and then $h\downarrow0$, we conclude that
\begin{equation} \label{2.13a}
 \lim_{n\to\infty}\limsup_{\ee\downarrow0}\limsup_{N\to\infty}
 \sup_{0\leq\beta\leq\beta_*}
 \Big\langle\Big\|2^q\mathcal G_q(\b t)-\prod_{k=1}^q 2\mathcal G_1(t_k)\Big\|_\infty\Big\rangle=0\,,
\end{equation}
which implies \eqref{2.13}. This finishes the proof.
\end{proof}

\subsection{The limiting equation} \label{sec2identification}

We shall prove Lemma \ref{lem2.4} using the following three results.

\begin{lemma} \label{lem2.8}
For $0\leq\beta\leq\beta_*$, we have $\mathcal G_1\in C^1([0,1])$ and
\begin{equation} \label{2.31}
\begin{aligned}
 \big\|\mathcal G_1'-\beta^2(\mathcal G_1^{q-1}*\mathcal G_1)\big\|_\infty\leq\frac{\beta_*^2}{2^q}\bigg(1-a+
 \bigg\|2^q\mathcal G_q(\b t)-\prod_{k=1}^q 2\mathcal G_1(t_k)\bigg\|_\infty\bigg)
 +\frac12\sum_{j=1}^n\gamma_j\,.
\end{aligned}
\end{equation}
In particular,
\begin{equation} \label{2.31a}
 \lim_{n\to\infty}\limsup_{\ee\downarrow0}\limsup_{N\to\infty}
 \sup_{0\leq\beta\leq\beta_*}
 \big\langle\|\mathcal G_1'-\beta^2(\mathcal G_1^{q-1}*\mathcal G_1)\|_\infty\big\rangle=0\,.
\end{equation}
\end{lemma}

\begin{lemma} \label{lem2.9}
For every $\beta\geq0$, there exists a normalized positive Lehmann solution of \eqref{1.3}. Every such solution belongs to $C^1([0,1])$ and satisfies \eqref{2.12}.
\end{lemma}

\begin{lemma} \label{lem2.11}
Let $0\leq\beta\leq\beta_*$, and let $G$ be a normalized positive Lehmann solution of \eqref{1.3}. For every normalized positive Lehmann kernel $F\in C^1([0,1])$, we have
\begin{equation} \label{2.33}
 \|F-G\|_\infty
 \leq\bigg(1+\frac{\beta_*^2 2^{1-q}}\pi\bigg)^{1/2}
 \|F'-\beta^2(F^{q-1}*F)\|_2\,.
\end{equation}
\end{lemma}

\begin{proof}[Proof of Lemma \ref{lem2.4}]
Let $G_\beta$ be a solution given by Lemma \ref{lem2.9}. Any other normalized positive Lehmann solution $F$ belongs to $C^1([0,1])$ and satisfies \eqref{2.12}, so \eqref{2.33} gives $F=G_\beta$. Thus the solution is unique. When $\beta=0$, \eqref{2.12} gives $G_0=1/2$.

By \eqref{2.21} and Lemma \ref{lem2.8}, we can take $F=\mathcal G_1$ in \eqref{2.33}. Averaging and using $\|\cdot\|_2\leq\|\cdot\|_\infty$ and \eqref{2.31a} yield \eqref{2.14a}. Taking $F=G_{\beta'}$ in \eqref{2.33}, and using \eqref{2.12} and $\|G_{\beta'}^{q-1}*G_{\beta'}\|_2\leq2^{-q}$, we get
\[
 \|G_{\beta'}-G_\beta\|_\infty
 \leq\bigg(1+\frac{\beta_*^2 2^{1-q}}\pi\bigg)^{1/2}2^{-q}|\beta'^2-\beta^2|\,,\quad
 0\leq\beta,\beta'\leq\beta_*\,.
\]
This proves the continuity of $\beta\mapsto G_\beta$ and finishes the proof.
\end{proof}

\begin{proof}[Proof of Lemma \ref{lem2.8}]
Fix $i\in\qq{N}$ and $0<t<1$. Taking $r=1$ and $i_1=i$ in \eqref{2.28a}, so that $\mathcal U_1=\psi_i$, and dividing by $2$, we get
\[
 \mathcal G_1'(t)-\frac\ii2\sum_{j=1}^n\gamma_j\Phi(\eta_{ij}(t)\psi_i)
 =-\beta^2\varkappa^2\sum_{A\ni i}\int_0^1
 \Phi(\mathcal T\Psi_A(s)(\Psi_A\psi_i)(t)\psi_i)\,\dd s\,.
\]
For $A\ni i$, the matrix $\Psi_A\psi_i$ is, up to a scalar factor, the product of the $q-1$ fields at sites in $A\setminus\{i\}$. By cyclicity of the trace, the equal-time identity in Section \ref{sec2.2}, and \eqref{2.19b}, we have
\[
\begin{aligned}
 \Phi(\mathcal T\Psi_A(s)(\Psi_A\psi_i)(t)\psi_i)
 &=\Phi(\psi_i(1-s)(\Psi_A\psi_i)(t-s)\Psi_A)\\
 &=-2^q\mathcal G_q(1-s,t-s,...,t-s)\,,\quad 0<s<t\,,
\end{aligned}
\]
and
\[
\begin{aligned}
 \Phi(\mathcal T\Psi_A(s)(\Psi_A\psi_i)(t)\psi_i)
 &=\Phi((\Psi_A\psi_i)(1+t-s)\psi_i(1-s)\Psi_A)\\
 &=2^q\mathcal G_q(1-s,1+t-s,...,1+t-s)\,,\quad t<s<1\,.
\end{aligned}
\]
The minus sign in the first case comes from moving $\psi_i(1-s)$ past the $q-1$ fields at time $t-s$. Splitting the integral at $s=t$, substituting $s=t+v$ in the second part, and using $2^q\varkappa^2\binom{N-1}{q-1}=a$, we obtain
\begin{equation} \label{2.29}
\begin{aligned}
 \mathcal G_1'(t)
 ={}&a\beta^2\bigg(\int_0^t\mathcal G_q(1-v,t-v,...,t-v)\,\dd v
 -\int_0^{1-t}\mathcal G_q(1-t-v,1-v,...,1-v)\,\dd v\bigg)\\
 &+\frac\ii2\sum_{j=1}^n\gamma_j\Phi(\eta_{ij}(t)\psi_i)\,.
\end{aligned}
\end{equation}

For a symmetric function $G(1-t)=G(t)$, the antiperiodic convolution is
\[
 (G^{q-1}*G)(t)=\int_0^tG(t-v)^{q-1}G(v)\,\dd v
 -\int_0^{1-t}G(v)^{q-1}G(t+v)\,\dd v\,.
\]
The right-hand side of \eqref{2.29} is continuous on $[0,1]$, so $\mathcal G_1\in C^1([0,1])$. Using $\mathcal G_1(1-t)=\mathcal G_1(t)$ in \eqref{2.29}, together with \eqref{2.19} and the convolution identity above, gives \eqref{2.31}.
By \eqref{2.13a}, $\lim_N a=1$, and $\sum_j\gamma_j\leq2n\delta$, \eqref{2.31} gives \eqref{2.31a}. This finishes the proof.
\end{proof}

\begin{proof}[Proof of Lemma \ref{lem2.9}]
Fix $\beta\geq0$ and choose $\beta_*>\beta$. By \eqref{2.31a}, for each integer $l\geq1$, first choose $n\geq l$, then $\ee\leq l^{-1}$, and finally $N\geq l$, so that
\begin{equation} \label{2.31b}
 \big\langle\|\mathcal G_1'-\beta^2(\mathcal G_1^{q-1}*\mathcal G_1)\|_\infty\big\rangle\leq l^{-1}\,,\quad
 2n\delta\leq l^{-1}\,.
\end{equation}
For these parameters choose $\b\gamma$ in the box in \eqref{2.10} such that the norm is no larger than its average.
Let $\mu$ be the probability measure on $[0,\infty)$ such that $\mathcal G_1=G_\mu$, as follows from \eqref{2.21}. Since $-k_x'(0)=x\tanh(x/2)/2$, Fatou's lemma, Lemma \ref{lem2.7a}, and \eqref{2.31b} give
\[
 \frac12\int_{[0,\infty)}x\tanh(x/2)\,\mu(\dd x)
 \leq\liminf_{t\downarrow0}\frac{1/2-\mathcal G_1(t)}t\leq C_{\beta,q}\,.
\]
Thus $\mu((x_*,\infty))\leq2C_{\beta,q}/(x_*\tanh(x_*/2))$ for every $x_*>0$, so the measures $\mu$ are tight. Passing to a subsequence, $\mu$ converges weakly to a probability measure $\mu_0$ on $[0,\infty)$, and hence $\mathcal G_1$ converges pointwise to $G\deq G_{\mu_0}$. By the Lipschitz bound in Lemma \ref{lem2.7a}, the convergence is uniform on $[0,1]$. Integrating $\mathcal G_1'$ on $[0,t]$, using \eqref{2.31b}, and passing to the limit give
\[
 G(t)=\frac12+\beta^2\int_0^t(G^{q-1}*G)(u)\,\dd u\,.
\]
Thus $G\in C^1([0,1])$ and \eqref{2.12} holds. With $\omega_k=(2k+1)\pi$ and $\wh f(k)=\int_0^1\e^{\ii\omega_k t}f(t)\,\dd t$, integration by parts gives
\begin{equation} \label{2.32}
 \wh{G'}(k)=-G(1)-G(0)-\ii\omega_k\wh G(k)
 =-1-\ii\omega_k\wh G(k)\,.
\end{equation}
Taking the Fourier transform of \eqref{2.12} and using \eqref{2.32} yields \eqref{1.3}.

Conversely, any normalized positive Lehmann solution $F$ of \eqref{1.3} belongs to $C^1([0,1])$ and satisfies \eqref{2.12}. Indeed, $F(1-t)=F(t)$ implies
\[
 (F^{q-1}*F)(1-t)=-(F^{q-1}*F)(t)\quad\mbox{and}\quad
 \int_0^1(F^{q-1}*F)(t)\,\dd t=0\,.
\]
Thus $\wt F(t)\deq1/2+\beta^2\int_0^t(F^{q-1}*F)(u)\,\dd u$ belongs to $C^1([0,1])$ and satisfies $\wt F(0)=\wt F(1)=1/2$. By \eqref{2.32} and \eqref{1.3}, we have
\[
 \wh{\wt F}(k)=-\frac{1+\beta^2\wh{F^{q-1}}(k)\wh F(k)}{\ii\omega_k}=\wh F(k)
\]
for all $k\in\bb Z$, which gives $F=\wt F$ and proves the claim.
\end{proof}

\begin{proof}[Proof of Lemma \ref{lem2.11}]
By Lemma \ref{lem2.9}, $G\in C^1([0,1])$ and satisfies \eqref{2.12}. Set $R\deq F'-\beta^2(F^{q-1}*F)$. Writing $F=G_\mu$, for $k\geq0$ we have
\[
 \wh F(k)=\ii\omega_k\int_{[0,\infty)}\frac{\mu(\dd x)}{\omega_k^2+x^2}\eqd\ii u_k\,,\quad 0<u_k\leq\omega_k^{-1}\,.
\]
Similarly, write $\wh G(k)\eqd\ii v_k$, and set
\[
 \sigma_k\deq-\ii\beta^2\wh{F^{q-1}}(k)\,,\quad
 \tau_k\deq-\ii\beta^2\wh{G^{q-1}}(k)\,.
\]
Since $F^{q-1}$ is symmetric about $1/2$ and decreasing on $[0,1/2]$, integration by parts gives
\[
 -\ii\wh{F^{q-1}}(k)=\frac2{\omega_k}
 \bigg(F(1/2)^{q-1}-\int_0^{1/2}(1-\cos(\omega_k t))(F^{q-1})'(t)\,\dd t\bigg)\geq0\,.
\]
The same holds for $G^{q-1}$, and thus $0\leq\sigma_k,\tau_k\leq\beta_*^2 2^{1-q}$. Set $d_k\deq u_k-v_k$. By \eqref{2.32} and the definitions of $R$ and $G$, we have $v_k^{-1}=\omega_k+\tau_k$ and
\[
 \wh R(k)=(\omega_k+\sigma_k)d_k+(\sigma_k-\tau_k)v_k\,.
\]
Multiplying by $d_k/v_k$ gives
\[
 \wh R(k)\frac{d_k}{v_k}
 =(\omega_k+\sigma_k)(\omega_k+\tau_k)d_k^2+d_k(\sigma_k-\tau_k)\,.
\]
The last term has a nonnegative sum. Indeed, $(F-G)(F^{q-1}-G^{q-1})\geq0$, so Parseval's identity gives
\[
 \sum_{k\geq0}d_k(\sigma_k-\tau_k)
 =\frac{\beta^2}{2}\int_0^1(F-G)(F^{q-1}-G^{q-1})\,\dd t\geq0\,.
\]
Since $F-G$ is $C^1$ and vanishes at both endpoints, Parseval's identity gives $\sum_{k\geq0}\omega_k^2d_k^2<\infty$. As $\sigma_k,\tau_k$ are bounded, all the series below converge. Summing the identity for $\wh R(k)d_k/v_k$, dropping the nonnegative sum, and applying Cauchy--Schwarz, we get
\[
\begin{aligned}
 \sum_{k\geq0}(\omega_k+\sigma_k)(\omega_k+\tau_k)d_k^2
 &\leq\sum_{k\geq0}(\omega_k+\tau_k)d_k\wh R(k)\\
 &\leq\bigg(\sum_{k\geq0}\frac{\omega_k+\tau_k}{\omega_k+\sigma_k}|\wh R(k)|^2\bigg)^{1/2}
 \bigg(\sum_{k\geq0}(\omega_k+\sigma_k)(\omega_k+\tau_k)d_k^2\bigg)^{1/2}\,.
\end{aligned}
\]
Thus we obtain
\[
 \|(F-G)'\|_2^2=2\sum_{k\geq0}\omega_k^2d_k^2
 \leq\bigg(1+\frac{\beta_*^2 2^{1-q}}\pi\bigg)2\sum_{k\geq0}|\wh R(k)|^2
 =\bigg(1+\frac{\beta_*^2 2^{1-q}}\pi\bigg)\|R\|_2^2\,.
\]
Together with $F(0)=G(0)$, this proves \eqref{2.33}.
\end{proof}

\begin{remark} \label{rem3.1}
The standard Schwinger--Dyson formula \cite[(2.26) and (G.238)]{MS16} expresses the pressure as $\mathcal S_\beta(G_\beta)$, where, in our normalization,
\[
 \mathcal S_\beta(F)\deq\sum_{k\geq0}\log\bigg(1-\frac{\ii\beta^2\wh{F^{q-1}}(k)}{\omega_k}\bigg)
 -\frac{(q-1)\beta^2}{2q}\int_0^1F(t)^q\,\dd t\,.
\]
With the ordinary trace, there is an additional term $(\log2)/2$. We prove that
\begin{equation} \label{2.34}
 p_{{\rm SD},q}(\beta)=\mathcal S_\beta(G_\beta)\,,
\end{equation}
which identifies our integral definition \eqref{pSD} with the standard formula.
For normalized positive Lehmann kernels $F\in C^1([0,1])$, the coefficients $-\ii\widehat{F^{q-1}}(k)$ are nonnegative, as shown in the proof of Lemma \ref{lem2.11}. By Parseval's identity and Cauchy--Schwarz, the series and its $\beta$-derivative converge uniformly on bounded $\beta$-intervals, uniformly in $F$. The inequality $\log y-\log x\leq(y-x)/x$, \eqref{1.3}, and Parseval's identity give
\[
 \mathcal S_\beta(F)-\mathcal S_\beta(G_\beta)
 \leq\frac{\beta^2}{2q}\int_0^1\big(qG_\beta F^{q-1}-(q-1)F^q-G_\beta^q\big)\,\dd t\leq0\,,
\]
where we used $qxy^{q-1}\leq x^q+(q-1)y^q$ for $x,y\geq0$. Applying this bound at $\beta$ and $\beta'$ yields
\[
 \mathcal S_{\beta'}(G_\beta)-\mathcal S_\beta(G_\beta)
 \leq\mathcal S_{\beta'}(G_{\beta'})-\mathcal S_\beta(G_\beta)
 \leq\mathcal S_{\beta'}(G_{\beta'})-\mathcal S_\beta(G_{\beta'})\,.
\]
Taking difference quotients and using the continuity of $\beta\mapsto G_\beta$, \eqref{1.3}, and Parseval's identity give
\[
\begin{aligned}
 \frac{\dd}{\dd\beta}\mathcal S_\beta(G_\beta)
 =2\beta\sum_{k\geq0}(-\ii\wh G_\beta(k))(-\ii\widehat{G_\beta^{q-1}}(k))
 -\frac{(q-1)\beta}{q}\int_0^1G_\beta^q\,\dd t=\frac\beta q\int_0^1G_\beta^q\,\dd t\,.
\end{aligned}
\]
Since $\mathcal S_0(G_0)=0$, integration and \eqref{pSD} give \eqref{2.34}.
\end{remark}

\section{The pressure comparison for growing \texorpdfstring{$q$}{q}} \label{secgrowthpressure}
In this section, we prove Proposition \ref{prop2.6}. Suppose $1\ll q\ll\sqrt N$, and set $\alpha\deq q^2/N$ and $s\deq\e^{-2\alpha}$. For the rest of the paper, we work with the original matrices of size $L=2^{N/2}$ and $\trn=L^{-1}\tr$. Recall $f(\beta)=\max_{0\leq v\leq\pi/2}(\beta\sin v-v^2)$ from \eqref{bbb}. By Lemma \ref{lem2.4a} and \eqref{pSD}, for every fixed $\beta_*>0$ and all large $N$, we have
\begin{equation} \label{2.18c}
 \sup_{0\leq\beta\leq\beta_*}\big|q^2p_{{\rm SD},q}(b\beta)-f(\beta)\big|
 \leq\frac{\beta_*^6}{12q}\,.
\end{equation}
Thus it suffices to prove that $q^2p_{\rm ann}(b\beta)\to f(\beta)$. Since $bH=\sqrt{a/(2\alpha)}\,\mathcal H$ by \eqref{1.5} and \eqref{bbb}, we have
\begin{equation} \label{4.18}
 q^2p_{\rm ann}(b\beta)=\alpha\log\bb E\trn\e^{\beta\sqrt a\,\mathcal H/\sqrt{2\alpha}}\,.
\end{equation}
We shall expand $\bb E\trn\e^{\beta\sqrt a\,\mathcal H/\sqrt{2\alpha}}$ and compare the moments up to order $k\asymp\alpha^{-1}$ with those of the $q$-Hermite law, where $kq/N=O(q^{-1})\to0$.

For $k\geq1$, Wick's theorem and \eqref{1.2} give
\begin{equation} \label{2.45}
 \bb E\trn\mathcal H^{2k}=\sum_\pi\bb E\prod_{\{i,j\}\in\mathcal E_\pi}(-1)^{|A_i\cap A_j|}\,,
\end{equation}
where $\pi=\{\{a_i,b_i\}:1\leq i\leq k\}$ runs over pairings of $\qq{2k}$, with $a_i<b_i$. Here $\{i,j\}\in\mathcal E_\pi$ if $a_i<a_j<b_i<b_j$ or $a_j<a_i<b_j<b_i$, and $A_1,...,A_k$ are i.i.d.\,uniform $q$-subsets of $\qq N$.

For $0\leq s<1$, the moment formula for the $q$-Hermite law (see e.g.\,\cite[(2.19)]{BINT}) gives
\begin{equation} \label{4.17}
 \int x^{2k}\varrho_s(x)\,\dd x=\sum_\pi s^{|\mathcal E_\pi|}\leq(2k-1)!!\,.
\end{equation}

\begin{lemma} \label{lem4.6}
If $N\geq256kq$, then we have
\[
 \e^{-Ck^3q^3/N^2}\int x^{2k}\varrho_s(x)\,\dd x
 \leq\bb E\trn\mathcal H^{2k}
 \leq\e^{Ck^3q^3/N^2}\int x^{2k}\varrho_s(x)\,\dd x\,.
\]
\end{lemma}
Recall that $L_s(t)=\log\int\e^{tx}\varrho_s(x)\,\dd x$.
\begin{lemma} \label{lem4.7}
For every fixed $\beta_*>0$ and all sufficiently small $\alpha>0$, we have
\begin{equation} \label{4.20}
 \sup_{0\leq\beta\leq\beta_*}\bigg|\alpha L_s\Big(\frac{\beta}{\sqrt{2\alpha}}\Big)-f(\beta)\bigg|
 \leq C_{\beta_*}\alpha\log(1/\alpha)\,.
\end{equation}
\end{lemma}

\begin{proof}[Proof of Proposition \ref{prop2.6}]
Fix $\beta_*>0$ and let $0\leq\beta\leq\beta_*$. Choose an integer $k_0\geq\e\beta_*^2+2$ and set $k_*\deq\ceil{k_0/\alpha}$. The odd moments vanish in both cases by symmetry, and we truncate the two exponential series at order $2k_*$. Since $k_*q/N\leq k_0/q+q/N\to0$, we have $N\geq256k_*q$ for all large $N$. By Lemma \ref{lem4.6}, for all $k\leq k_*$ we have
\[
 \bb E\trn\mathcal H^{2k}=\e^{O(k_*^3q^3/N^2)}\int x^{2k}\varrho_s(x)\,\dd x\,,\quad
 \mbox{where}\quad \frac{k_*^3q^3}{N^2}\leq \frac{C_{k_0}}{\alpha q}\,.
\]
By \eqref{2.45} and \eqref{4.17}, the two moments $\bb E\trn\mathcal H^{2m}$ and $\int x^{2m}\varrho_s(x)\,\dd x$ are in $[0,(2m-1)!!]$. Since $(2m-1)!!/(2m)!=1/(2^m m!)$, the tails of both exponential series are bounded by
\[
 \sum_{m>k_*}\frac{(\beta^2/(4\alpha))^m}{m!}
 \leq2\Big(\frac{\e\beta^2}{4\alpha(k_*+1)}\Big)^{k_*+1}
 \leq2\e^{-k_0\log4/\alpha}\,.
\]
By Lemma \ref{lem4.6} and the above tail estimate, we have
\[
\begin{aligned}
 \bb E\trn\e^{\beta\mathcal H/\sqrt{2\alpha}}
 \leq \e^{C_{\beta_*}/(\alpha q)+L_s(\beta/\sqrt{2\alpha})}
 +\frac{2}{4^{k_0/\alpha}}\,,\ 
 \bb E\trn\e^{\beta\mathcal H/\sqrt{2\alpha}}
 \geq \e^{-C_{\beta_*}/(\alpha q)}
 \Big(\e^{L_s(\beta/\sqrt{2\alpha})}
 -\frac{2}{4^{k_0/\alpha}}\Big)\,.
\end{aligned}
\]
Since $\e^{L_s(\beta/\sqrt{2\alpha})}\geq1$, taking logarithms gives
\begin{equation} \label{4.15}
 \sup_{0\leq\beta\leq\beta_*}\alpha\bigg|\log\bb E\trn\e^{\beta\mathcal H/\sqrt{2\alpha}}
 -L_s\Big(\frac{\beta}{\sqrt{2\alpha}}\Big)\bigg|
 \leq\frac{C_{\beta_*}}q+C\alpha\e^{-k_0\log4/\alpha}\,.
\end{equation}

Applying \eqref{4.15} and \eqref{4.20} at $\beta\sqrt a$ and using \eqref{4.18}, $1-\sqrt a\leq q(q-1)/(2N)$, and $|f'|\leq1$, we get
\[
 \sup_{0\leq\beta\leq\beta_*}|q^2p_{\rm ann}(b\beta)-f(\beta)|
 \leq C_{\beta_*}\big(q^{-1}+\alpha\log(1/\alpha)\big)\to0\,.
\]
Together with \eqref{2.18c}, this concludes the proof of Proposition \ref{prop2.6}.
\end{proof}

\subsection{Comparison of the moments}
It remains to prove Lemma \ref{lem4.6}.

\begin{proof}
(i) For a simple graph on $\qq m$ with edge set $\mathcal E$, set
\[
 w(\mathcal E)\deq\bb E\prod_{\{i,j\}\in\mathcal E}(-1)^{|A_i\cap A_j|}\,,
\]
where $A_1,...,A_m$ are i.i.d.\,uniform $q$-subsets of $\qq N$. By \eqref{2.45} and \eqref{4.17}, it suffices to show that for $N\geq256mq$,
\begin{equation} \label{4.19}
 w(\mathcal E)>0\quad\mbox{and}\quad
 \big|\log w(\mathcal E)+2\alpha|\mathcal E|\big|\leq C\frac{m^3q^3}{N^2}\,.
\end{equation}
Partition $\qq{mq}$ into $m$ blocks $B_1,...,B_m$ of size $q$. For distinct $x,y\in\qq{mq}$, set $v_{xy}\deq -1$ if $x,y$ belong to the same block $B_i$, $v_{xy}\deq -2$ if $x\in B_i$ and $y\in B_j$ with $\{i,j\}\in\mathcal E$, and $v_{xy}\deq 0$ otherwise. Note that $|v_{xy}|\leq2$, $|1+v_{xy}|\leq1$, and
\[
 \sum_{1\leq x<y\leq mq}v_{xy}=-\frac{mq(q-1)}2-2q^2|\mathcal E|\,.
\]
Let $\sigma_1,...,\sigma_{mq}$ be independent and uniformly distributed on $\qq N$. For $S\subset\qq{mq}$, define
\[
 z_S\deq\bb E\prod_{\substack{x,y\in S\\x<y}}\big(1+v_{xy}\mathbf{1}_{\sigma_x=\sigma_y}\big)\,,\quad z_\emptyset\deq1\,.
\]
When $S=\qq{mq}$, the factors $1-\mathbf{1}_{\sigma_x=\sigma_y}$ within each block $B_i$ vanish unless the $q$ elements $(\sigma_x)_{x\in B_i}$ are distinct, which occurs with probability $(N)_q/N^q$. Conditioned on this event for all $i\in\qq{m}$, the sets $A_i\deq\{\sigma_x:x\in B_i\}$ are i.i.d.\,uniform $q$-subsets of $\qq{N}$, and $\prod_{x\in B_i,\,y\in B_j}(1-2\mathbf{1}_{\sigma_x=\sigma_y})=(-1)^{|A_i\cap A_j|}$ for $\{i,j\}\in\mathcal E$. As a result, we have
\begin{equation} \label{z_to_w}
 z_{\qq{mq}}=\Big(\frac{(N)_q}{N^q}\Big)^m w(\mathcal E)\,.
\end{equation}
Since $q/N\leq1/256$, we have
\[
 m\log\frac{(N)_q}{N^q}=-\frac{mq(q-1)}{2N}+O\Big(\frac{mq^3}{N^2}\Big)\,.
\]
The leading term cancels the contribution from pairs in the same block in $N^{-1}\sum_{x<y}v_{xy}$. Thus, by \eqref{z_to_w}, \eqref{4.19} follows if we show that
\begin{equation} \label{4.12}
 z_{\qq{mq}}>0\quad\mbox{and}\quad
 \bigg|\log z_{\qq{mq}}-\frac1N\sum_{1\leq x<y\leq mq}v_{xy}\bigg|\leq C\frac{m^3q^3}{N^2}\,.
\end{equation}

(ii) To prove \eqref{4.12}, we first obtain a recursion for $z_S$. For $S\subset\qq{mq}$, expanding the product in the definition of $z_S$ over subgraphs $\mathcal E'\subset\binom S2$ gives
\[
 z_S=\sum_{\mathcal E'\subset\binom S2}\bigg(\prod_{e\in\mathcal E'}v_e\bigg)\bb E\prod_{\{y,w\}\in\mathcal E'}\mathbf{1}_{\sigma_y=\sigma_w}\,.
\]
Fix $x\in S$, and group the subgraphs by the vertex set $T$ of the component containing $x$. For a connected graph $(T,\mathcal F')$, we have
\[
 \bb E\prod_{\{y,w\}\in\mathcal F'}\mathbf1_{\sigma_y=\sigma_w}
 =\bb P(\sigma_y=\sigma_x\ \mbox{for all }y\in T)=N^{1-|T|}\,.
\]
By independence, summing the subgraphs on $S\setminus T$ gives $z_{S\setminus T}$. Including $T=\{x\}$, we obtain
\begin{equation} \label{4.13}
 z_S=z_{S\setminus\{x\}}+
 \sum_{\substack{T\subset S,\,x\in T\\|T|\geq2}}\frac{\vartheta_T}{N^{|T|-1}}z_{S\setminus T}\,,
 \quad \mbox{where} \quad
 \vartheta_T\deq\sum_{\substack{\mathcal F'\subset\binom T2\\(T,\mathcal F')\mbox{\scriptsize\ connected}}}\prod_{e\in\mathcal F'}v_e\,.
\end{equation}
Fix a linear ordering of $\binom T2$. For a tree $\mathcal F$ on $T$, let $\mathcal E_{\mathcal F}$ consist of the edges outside $\mathcal F$ that are larger than every edge on the path in $\mathcal F$ between their endpoints. Adding edges in increasing order and discarding edges that create a cycle gives $\mathcal F$ precisely for the connected graphs with edge set $\mathcal F'$ satisfying $\mathcal F\subset\mathcal F'\subset\mathcal F\cup\mathcal E_{\mathcal F}$. Hence
\[
 \vartheta_T=\sum_{\mathcal F\mbox{\scriptsize\ tree on }T}\prod_{e\in\mathcal F}v_e\prod_{e\in \mathcal E_{\mathcal F}}(1+v_e)\,.
\]
Rooting a tree at a fixed vertex identifies it by the parents of its other $|T|-1$ vertices, so there are at most $|T|^{|T|-1}$ trees. Thus $|v_e|\leq2$ and $|1+v_e|\leq1$ give $|\vartheta_T|\leq2^{|T|-1}|T|^{|T|-1}$.

(iii) To estimate $\log z_{\qq{mq}}$ from \eqref{4.13}, we first prove by induction on $|S|$ that for all $x\in S$,
\begin{equation} \label{4.14}
 z_S>0\quad\mbox{and}\quad
 \Big|\frac{z_S}{z_{S\setminus\{x\}}}-1\Big|\leq\frac{32mq}{N}\leq\frac18\,.
\end{equation}
This is trivial for $|S|\leq1$. For the induction step, removing the $|T|-1$ vertices of $T\setminus\{x\}$ one by one and applying the induction hypothesis give $|z_{S\setminus T}/z_{S\setminus\{x\}}|\leq(8/7)^{|T|-1}\leq 2^{|T|-1}$. Using \eqref{4.13}, the bound on $|\vartheta_T|$, and $k!\geq(k/\e)^k$, we obtain
\[
 \Big|\frac{z_S}{z_{S\setminus\{x\}}}-1\Big|
 \leq\sum_{k\geq1}\binom{mq-1}{k}\frac{4^k(k+1)^k}{N^k}
 \leq\sum_{k\geq1}\Big(\frac{8\e mq}{N}\Big)^k
 \leq\frac{32mq}{N}\,.
\]
Together with $z_{S\setminus\{x\}}>0$, this also gives $z_S>0$, which completes the induction.

Note that for $T=\{x,y\}$, we have $\vartheta_T=v_{xy}$ and $z_{S\setminus\{x,y\}}/z_{S\setminus\{x\}}=1+O(mq/N)$ by \eqref{4.14}. Moreover, the sum over $|T|\geq3$ in \eqref{4.13} is bounded by the same geometric series as above starting at $k=2$. Thus \eqref{4.13} gives
\[
 \frac{z_S}{z_{S\setminus\{x\}}}=1+\frac1N\sum_{y\in S\setminus\{x\}}v_{xy}
 +O\Big(\frac{m^2q^2}{N^2}\Big)\,,\quad
 \log\frac{z_S}{z_{S\setminus\{x\}}}=\frac1N\sum_{y\in S\setminus\{x\}}v_{xy}
 +O\Big(\frac{m^2q^2}{N^2}\Big)\,.
\]
Taking $S=\qq j$ and $x=j$ in the last relation, and summing over $1\leq j\leq mq$, we get
\[
 \log z_{\qq{mq}}=\sum_{j=1}^{mq}\log\frac{z_{\qq j}}{z_{\qq{j-1}}}
 =\frac1N\sum_{1\leq x<y\leq mq}v_{xy}+O\Big(\frac{m^3q^3}{N^2}\Big)\,,
\]
where $\qq0\deq\emptyset$ and $z_{\qq0}=1$. This proves \eqref{4.12}. Together with \eqref{z_to_w}, we obtain \eqref{4.19} as desired. This finishes the proof.
\end{proof}

\subsection{The \texorpdfstring{$q$}{q}-Hermite pressure}
\begin{proof}[Proof of Lemma \ref{lem4.7}]
Fix $\beta_*>0$ and let $0\leq\beta\leq\beta_*$. Under the substitution $x=2\cos u/\sqrt{1-s}$ for $u\in (0,\pi)$, \eqref{qHermite} becomes $\varrho_s(x)|\dd x|=h_\alpha(u)\,\dd u$, where
\[
 h_\alpha(u)\deq\frac1{2\pi}\prod_{k=1}^\infty(1-\e^{-2\alpha k})
 \prod_{k=0}^\infty\big(1-2\e^{-2\alpha k}\cos(2u)+\e^{-4\alpha k}\big)\,.
\]
Since $\sqrt{2\alpha/(1-\e^{-2\alpha})}=1+O(\alpha)$ and $|\cos u|\leq1$, we have
\[
 \bigg|\alpha L_s\Big(\frac{\beta}{\sqrt{2\alpha}}\Big)
 -\alpha\log\int_0^\pi\e^{\beta\cos u/\alpha}h_\alpha(u)\,\dd u\bigg|
 \leq C_{\beta_*}\alpha\,.
\]
We claim that for every fixed $\eta>0$,
\begin{equation} \label{4.16}
\begin{aligned}
 \alpha\log h_\alpha(u)&=-(u-\pi/2)^2+O_{\eta}\big(\alpha\log(1/\alpha)\big)\,,\quad u\in [\eta,\pi-\eta]\,,\\
 \alpha\log h_\alpha(u)&\leq -(u-\pi/2)^2+C\alpha\log(1/\alpha)\,,\quad u\in (0,\pi)\,.
\end{aligned}
\end{equation}

To prove \eqref{4.16}, set $\varphi_u(y)\deq\log(1-2\e^{-y}\cos(2u)+\e^{-2y})$ for $u\in(0,\pi)$ and $y\geq0$. Expanding at $y>0$ as $\varphi_u(y)=-2\sum_{j=1}^\infty j^{-1}\e^{-jy}\cos(2ju)$ and integrating over $y\in(0,\infty)$ give
\[
 \int_0^\infty \varphi_u(y)\,\dd y=-2\sum_{j=1}^\infty\frac{\cos(2ju)}{j^2}
 =-2\big(\pi^2/6-\pi u+u^2\big)\,.
\]
Moreover, differentiating in $y$ gives
\[
 \varphi_u'(y)=\frac{2\e^{-y}(\cos(2u)-\e^{-y})}{1-2\e^{-y}\cos(2u)+\e^{-2y}}\,.
\]
The sign of this derivative gives
\[
 \int_0^\infty(-\varphi_u'(y))_+\,\dd y
 =\begin{cases}
  \log(2-2\cos(2u)) & \mbox{if }\cos(2u)\leq0\,,\\
  \log\dfrac{2}{1+\cos(2u)} & \mbox{if }\cos(2u)>0
 \end{cases}
 \leq\log4\,.
\]
Comparing the sum with the integral on each interval $[2\alpha k,2\alpha(k+1)]$ gives
\[
 \sum_{k=0}^\infty \varphi_u(2\alpha k)-\frac1{2\alpha}\int_0^\infty \varphi_u(y)\,\dd y
 \leq\int_0^\infty(-\varphi_u'(y))_+\,\dd y\leq\log4\,,
\]
while for $u\in[\eta,\pi-\eta]$ the absolute value of the LHS is bounded by $\int_0^\infty|\varphi_u'(y)|\,\dd y\leq C_\eta$. Similarly, since $y\mapsto\log(1-\e^{-y})$ is increasing on $(0,\infty)$ and $\int_0^\infty\log(1-\e^{-y})\,\dd y=-\pi^2/6$, we have
\[
 0\leq\sum_{k=1}^\infty\log(1-\e^{-2\alpha k})
 +\frac{\pi^2}{12\alpha}
 \leq-\frac1{2\alpha}\int_0^{2\alpha}\log(1-\e^{-y})\,\dd y
 \leq C\log(1/\alpha)\,.
\]
Combining the above estimates proves \eqref{4.16}.

Note that $\max_{u\in[0,\pi]}(\beta\cos u-(u-\pi/2)^2)=f(\beta)$, where the unique maximizer is $u_*=\pi/2-v\in (0,\pi/2]$ with $\beta\cos v=2v$. For $0\leq\beta\leq\beta_*$, we have $u_*\in[2\eta_*,\pi/2]$ for some fixed $\eta_*>0$.
For $\alpha\leq\eta_*$, the interval $[u_*-\alpha,u_*+\alpha]$ lies in $[\eta_*,\pi-\eta_*]$, and on this interval we have $\beta\cos u-(u-\pi/2)^2\geq f(\beta)-C_{\beta_*}\alpha$. Using the upper bound of \eqref{4.16} on $(0,\pi)$ and the lower bound on $[u_*-\alpha,u_*+\alpha]$, we get
\[
 2\alpha\e^{f(\beta)/\alpha-C_{\beta_*}\log(1/\alpha)}
 \leq\e^{L_s(\beta/\sqrt{2\alpha})}
 \leq\pi\e^{f(\beta)/\alpha+C_{\beta_*}\log(1/\alpha)}\,.
\]
Taking $\alpha$ times the logarithm proves \eqref{4.20}. This finishes the proof.
\end{proof}

\section{The spectral edge for fixed \texorpdfstring{$q$}{q}} \label{sec3}
In this section, we prove Theorem \ref{mainthm1}\,(i). Fix an even integer $q\geq 4$, and abbreviate $G_\beta\equiv G_{\beta,q}$.

\begin{lemma} \label{lem3.1}
	For every $\beta>0$, we have
	\begin{equation} \label{3.1}
		\frac{1}{N\beta}\log\trn \e^{\beta H}\leq \frac{\lambda_{\max}(H)}{N}\leq \frac{1}{N\beta}\log\trn \e^{\beta H}+\frac{\log 2}{2\beta}\,,
	\end{equation}
	and thus
	\begin{equation} \label{3.2}
		\frac{p_{\rm que}(\beta)}{\beta}\leq \frac{\bb E\lambda_{\max}(H)}{N}\leq \frac{p_{\rm que}(\beta)}{\beta}+\frac{\log 2}{2\beta}\,.
	\end{equation}
	In addition, we have
	\[
		\lim_{N\to\infty}\frac{\bb E\lambda_{\max}(H)}{N}=\lim_{\beta\to\infty}\frac{p_{{\rm SD},q}(\beta)}{\beta}=e_{{\rm SD},q}\in [0,\infty)\,.
	\]
\end{lemma}
\begin{proof}
The first two estimates follow from $L^{-1}\e^{\beta\lambda_{\max}(H)}\leq\trn\e^{\beta H}\leq\e^{\beta\lambda_{\max}(H)}$ and taking expectation. Taking $N\to\infty$ and applying Theorem \ref{mainthm2}\,(i), we get
\[
	\frac{p_{{\rm SD},q}(\beta)}{\beta}\leq \liminf_{N\to\infty}\frac{\bb E\lambda_{\max}(H)}{N}\leq \limsup_{N\to\infty}\frac{\bb E\lambda_{\max}(H)}{N}\leq \frac{p_{{\rm SD},q}(\beta)}{\beta}+\frac{\log 2}{2\beta}\,.
\]
The liminf and limsup are finite by the above at $\beta=1$. Their difference is at most $(\log2)/(2\beta)$ for every $\beta>0$, so they agree. The same inequalities show that $p_{{\rm SD},q}(\beta)/\beta$ has the same limit. Nonnegativity follows from \eqref{pSD}.
\end{proof}

\begin{proof}[Proof of \eqref{jjj}]
As $\lambda_{\max}(H)=d\lambda_1$ and $\lim_N a=1$, Lemma \ref{lem3.1} gives
\[
	\frac{\bb E\lambda_1}{\sqrt{N}}=\sqrt{\frac{q2^q}{a}}\,\frac{\bb E\lambda_{\max}(H)}{N}\to \kappa_{{\rm SD},q}\,.
\]
By \eqref{2.4}, for any fixed $\ee>0$ we have
\[
	\bb P\big(|\lambda_1-\bb E\lambda_1|>\ee\sqrt{N}\big)\leq 2\e^{-N(N-q+1)\ee^2/2}\,.
\]
The Borel--Cantelli lemma gives \eqref{jjj}.
\end{proof}

\subsection{The zero-temperature Schwinger--Dyson equation} \label{sec3.1}
To prove \eqref{kkk}, we first study the limit of \eqref{1.3} as $\beta\to\infty$. For $\beta>0$, let $\rho_\beta$ be the probability measure on $[0,\infty)$ obtained by rescaling $x=\beta E$ in the Lehmann representation of $G_\beta$, so that
\begin{equation} \label{3.3}
	G_\beta(t)=\int_{[0,\infty)}\frac{\e^{-\beta t E}+\e^{-\beta(1-t)E}}{2(1+\e^{-\beta E})}\,\rho_\beta(\dd E)\,,\quad 0\leq t\leq 1\,.
\end{equation}
For a finite measure $\rho$ on $[0,\infty)$, define
\[
	m_\rho(v)\deq \int_{[0,\infty)}\frac{v}{v^2+E^2}\,\rho(\dd E)\,,\quad v>0\,.
\]
For a probability measure $\rho$ and $h\in\bb N_+$, denote by $\rho^{*h}$ the distribution of the sum of $h$ independent random variables with distribution $\rho$.

\begin{lemma} \label{lem3.2}
	For every $\beta>0$, $\rho_\beta$ satisfies
	\begin{equation} \label{3.4}
		\int_{[0,\infty)} E^2\,\rho_\beta(\dd E)=2^{2-q}\,.
	\end{equation}
	Every weak limit point $\rho_0$ of $\rho_\beta$ as $\beta\to\infty$ is a probability measure on $[0,\infty)$ satisfying
	\begin{equation} \label{3.5}
		m_{\rho_0}(v)=\frac{1}{v+2^{2-q}m_{\rho_0^{*(q-1)}}(v)}\,,\quad v>0\,.
	\end{equation}
	In addition, we have
	\begin{equation} \label{3.6}
		\int_{[0,\infty)} E^2\,\rho_0(\dd E)=2^{2-q}\,,\quad \rho_0(\{0\})=0\,,\quad \mbox{and} \quad \int_{[0,\infty)} E\,\rho_0(\dd E)=2q e_{{\rm SD},q}>0\,.
	\end{equation}
\end{lemma}
\begin{proof}
Set $h\deq q-1$. To write \eqref{1.3} in terms of $\rho_\beta$, we first find a Lehmann representation of $G_\beta^h$. Let $E_1,...,E_h$ be i.i.d.\,random variables with distribution $\rho_\beta$. Given $E_1,...,E_h$, let $\sigma_1,...,\sigma_h\in\{\pm1\}$ be conditionally independent with
\[
\bb P(\sigma_j=\pm 1\mid E_j)\deq\frac{\e^{\pm\beta E_j/2}}{\e^{\beta E_j/2}+\e^{-\beta E_j/2}}\,,
\]
and set $e_j\deq\sigma_j E_j$. Then $\bb E\e^{-\beta t e_j}=2G_\beta(t)$, and flipping the signs of $\sigma_1,...,\sigma_h$ gives
\[
\bb P\Big(\sum_{j=1}^h e_j\in-\dd E\Big)=\e^{-\beta E}\bb P\Big(\sum_{j=1}^h e_j\in\dd E\Big)\,.
\]
Let $\nu_\beta$ be $2^{2-q}$ times the distribution of $|e_1+\cdots+e_h|$. The above relations give
\[
	G_\beta(t)^h=2^{-h}\bb E\e^{-\beta t\sum_{j=1}^h e_j}=\int_{[0,\infty)} k_{\beta E}(t)\,\nu_\beta(\dd E)\,,\quad \nu_\beta([0,\infty))=2^{2-q}\,.
\]
Recall $\omega_k=(2k+1)\pi$. Since $\wh k_{\beta E}(k)=\ii\omega_k/(\omega_k^2+\beta^2 E^2)$, for $k\geq 0$ we get
\[
	\wh G_\beta(k)=\frac{\ii}{\beta}m_{\rho_\beta}(\omega_k/\beta)\quad \mbox{and} \quad \wh{G_\beta^h}(k)=\frac{\ii}{\beta}m_{\nu_\beta}(\omega_k/\beta)\,,
\]
so \eqref{1.3} becomes
\begin{equation} \label{3.7}
	m_{\rho_\beta}(v)=\frac{1}{v+m_{\nu_\beta}(v)}\,,\quad v=\frac{\omega_k}{\beta}\,.
\end{equation}
Since $1-vm_{\rho_\beta}(v)=\int_{[0,\infty)} E^2/(v^2+E^2)\,\rho_\beta(\dd E)$, \eqref{3.7} gives
\begin{equation} \label{3.8}
	\int_{[0,\infty)}\frac{v^2 E^2}{v^2+E^2}\,\rho_\beta(\dd E)=\frac{vm_{\nu_\beta}(v)}{1+m_{\nu_\beta}(v)/v}\,.
\end{equation}
Letting $k\to\infty$, we have $v\to\infty$ and $vm_{\nu_\beta}(v)\to 2^{2-q}$, so monotone convergence on the LHS of \eqref{3.8} yields \eqref{3.4}.

To pass to the limit in \eqref{3.7}, we compare $m_{\nu_\beta}$ and $m_{\rho_\beta^{*h}}$. Since $|e_j|=E_j$, we have
\[
	0\leq \sum_{j=1}^h E_j-\Big|\sum_{j=1}^h e_j\Big|\leq 2\sum_{j=1}^h E_j\mathbf{1}_{e_j<0}\,,\quad \bb E(E_j\mathbf{1}_{e_j<0}\mid E_j)=\frac{E_j}{1+\e^{\beta E_j}}\leq\frac{1}{\e\beta}\,.
\]
As $E\mapsto v/(v^2+E^2)$ is decreasing on $[0,\infty)$ with derivative bounded in absolute value by $v^{-2}$, we obtain
\begin{equation} \label{3.9}
	0\leq m_{\nu_\beta}(v)-2^{2-q}m_{\rho_\beta^{*h}}(v)\leq \frac{2^{3-q}}{v^2}\sum_{j=1}^h \int_{[0,\infty)}\frac{E}{1+\e^{\beta E}}\,\rho_\beta(\dd E)\leq \frac{2h2^{2-q}}{\e\beta v^2}\,.
\end{equation}
For any fixed $v>0$, choose $v_\beta=\omega_k/\beta$ such that $|v-v_\beta|\leq \pi/\beta$. Since $|m_\rho'(u)|\leq u^{-2}\rho([0,\infty))$, we get from \eqref{3.7} and \eqref{3.9} that
\begin{equation} \label{3.10}
	m_{\rho_\beta}(v)-\frac{1}{v+2^{2-q}m_{\rho_\beta^{*h}}(v)}=O_{q,v}(\beta^{-1})\,.
\end{equation}
By \eqref{3.4}, the family $\{\rho_\beta:\beta>0\}$ is tight. Along any weakly convergent subsequence $\rho_{\beta_l}\to\rho_0$ with $\beta_l\to\infty$, we also have $\rho_{\beta_l}^{*h}\to\rho_0^{*h}$ weakly. Since $E\mapsto v/(v^2+E^2)$ is bounded and continuous for $v>0$, passing to the limit in \eqref{3.10} gives \eqref{3.5}.

Repeating the derivation of \eqref{3.8} for \eqref{3.5} with $2^{2-q}\rho_0^{*h}$ in place of $\nu_\beta$ and taking $v\to\infty$ give $\int_{[0,\infty)} E^2\,\rho_0(\dd E)=2^{2-q}$. In addition, multiplying \eqref{3.5} by $v$ gives
\[
	\big(vm_{\rho_0}(v)\big)\big(v^2+2^{2-q}vm_{\rho_0^{*h}}(v)\big)=v^2\,.
\]
Since $\lim_{v\downarrow0}vm_\rho(v)=\rho(\{0\})$ and $\rho_0^{*h}(\{0\})=\rho_0(\{0\})^h$, taking $v\downarrow0$ yields $2^{2-q}\rho_0(\{0\})^{h+1}=0$, and thus $\rho_0(\{0\})=0$.

It remains to prove the first moment identity in \eqref{3.6}. Since $p_{\rm ann}(\beta)$ is convex, Theorem \ref{mainthm2}\,(i) implies that $p_{{\rm SD},q}(\beta)$ is convex. Thus
\[
 \frac{p_{{\rm SD},q}(\beta)-p_{{\rm SD},q}(\beta/2)}{\beta/2}
 \leq p_{{\rm SD},q}'(\beta)
 \leq\frac{p_{{\rm SD},q}(2\beta)-p_{{\rm SD},q}(\beta)}\beta\,.
\]
Both sides tend to $e_{{\rm SD},q}$ by Lemma \ref{lem3.1}, so $\lim_{\beta\to\infty}p_{{\rm SD},q}'(\beta)=e_{{\rm SD},q}$.

By antiperiodicity and $G_\beta(1-u)=G_\beta(u)$, we have $(G_\beta^h*G_\beta)(0+)=-\int_0^1 G_\beta(u)^q\,\dd u$. Since $|\partial_t k_{\beta E}(t)|\leq\beta E/2$ and \eqref{3.4} gives $\int E\,\rho_\beta(\dd E)<\infty$, we may differentiate \eqref{3.3} at $t=0+$. Together with \eqref{2.12} and \eqref{pSD}, this gives
\begin{equation} \label{3.11}
	qp_{{\rm SD},q}'(\beta)=-\frac{G_\beta'(0+)}{\beta}=\frac{1}{2}\int_{[0,\infty)} E\tanh(\beta E/2)\,\rho_\beta(\dd E)\,.
\end{equation}
Let $\rho_{\beta_l}\to \rho_0$ and fix $x_*>0$. Since $0\leq E(1-\tanh(\beta E/2))\leq 2/(\e\beta)$, the second moment bounds \eqref{3.4} and \eqref{3.6} give
\[
	\bigg|2qp_{{\rm SD},q}'(\beta_l)-\int_{[0,\infty)} E\,\rho_0(\dd E)\bigg|\leq \frac{2}{\e\beta_l}+\frac{2^{3-q}}{x_*}+\bigg|\int_{[0,\infty)}\min(E,x_*)\,(\rho_{\beta_l}-\rho_0)(\dd E)\bigg|\,.
\]
Taking $\beta_l\to\infty$ and then $x_*\to\infty$ gives $\int_{[0,\infty)} E\,\rho_0(\dd E)=2q e_{{\rm SD},q}>0$. This finishes the proof.
\end{proof}

\subsection{Bounds on the edge constant} \label{sec3.2}

\begin{proof}[Proof of \eqref{kkk}]
Abbreviate $\rho\equiv \rho_0$ and $\nu\equiv \rho^{*(q-1)}$, where $\rho_0$ is as in Lemma \ref{lem3.2}. Set $\Lambda(x)\deq \int_{[0,\infty)}\frac{E}{E+x}\,\rho(\dd E)$. We shall prove \eqref{kkk} by showing that
\begin{equation} \label{3.15}
	2^{2-q}-\Big(\int_{[0,\infty)} E\,\rho(\dd E)\Big)^2=2^{2-q}\int_{[0,\infty)}(1-\Lambda(x))^2\,\nu(\dd x)\,.
\end{equation}

We first deduce \eqref{kkk} from \eqref{3.15}. Let $E_1,...,E_q$ be i.i.d.\,random variables with distribution $\rho$. By symmetry,
\[
	\int_{[0,\infty)}\Lambda(x)\,\nu(\dd x)=\bb E\frac{E_1}{E_1+\cdots+E_q}=\frac{1}{q}\,.
\]
Inserting this along with \eqref{3.6} and \eqref{1.10} into \eqref{3.15} gives
\[
	q^2\kappa_{{\rm SD},q}^2=2-q\int_{[0,\infty)}\Lambda(x)^2\,\nu(\dd x)\,.
\]
By Cauchy--Schwarz, $\int\Lambda(x)^2\,\nu(\dd x)\geq(\int\Lambda(x)\,\nu(\dd x))^2=q^{-2}$, which gives the upper bound $q^2\kappa_{{\rm SD},q}^2\leq2-q^{-1}$.

For the lower bound, let $E_1,...,E_{q+1}$ be i.i.d.\,with distribution $\rho$, and set $r_i\deq E_i/(E_1+\cdots+E_{q+1})$. For $i\neq j$, the sum $x=\sum_{k\notin\{i,j\}}E_k$ has distribution $\nu$ and is independent of $E_i,E_j$. Thus
\[
	\bb E\Big(\frac{r_i r_j}{(1-r_i)(1-r_j)}\,\Big|\,x\Big)=\bb E\Big(\frac{E_i}{E_i+x}\frac{E_j}{E_j+x}\,\Big|\,x\Big)=\Lambda(x)^2\,.
\]
Since $r_i+r_j\leq 1$, we have
\[
	\frac{1}{(1-r_i)(1-r_j)}\leq \frac{1}{1-r_i}+\frac{1}{1-r_j}\quad \mbox{and} \quad \sum_{i\neq j}\frac{r_i r_j}{1-r_i}=\sum_{i=1}^{q+1} r_i=1\,.
\]
Taking expectation and summing over $i\neq j$ give
\[
 q(q+1)\int_{[0,\infty)}\Lambda(x)^2\,\nu(\dd x)
 =\bb E\sum_{i\neq j}\frac{r_i r_j}{(1-r_i)(1-r_j)}\leq2\,.
\]
Hence $q^2\kappa_{{\rm SD},q}^2\geq2-2/(q+1)=2q/(q+1)$, which proves the lower bound in \eqref{kkk}.

To prove \eqref{3.15}, let $E$ have distribution $\rho$ and set $\zeta_v\deq E^2/(v^2+E^2)$ for $v>0$. Using $v/m_\rho(v)=v^2+2^{2-q}vm_\nu(v)$ from \eqref{3.5} and $\mathrm{Cov}(\zeta_v,\zeta_u)=\mathrm{Cov}(1-\zeta_v,1-\zeta_u)$, we get for $u\neq v$ that
\begin{equation} \label{3.14}
\begin{aligned}
 \mathrm{Cov}(\zeta_v,\zeta_u)
 &=\frac{u^2v^2}{u^2-v^2}\bigg(\frac{m_\rho(v)}v-\frac{m_\rho(u)}u\bigg)-uvm_\rho(v)m_\rho(u)\\
 &=2^{2-q}uvm_\rho(v)m_\rho(u)\frac{vm_\nu(v)-um_\nu(u)}{v^2-u^2}\\
 &=2^{2-q}\int_{[0,\infty)}x^2\frac{vm_\rho(v)}{v^2+x^2}\frac{um_\rho(u)}{u^2+x^2}\,\nu(\dd x)\geq0\,.
\end{aligned}
\end{equation}
Since $E=(2/\pi)\int_0^\infty\zeta_v\,\dd v$ and $\bb E E^2=2^{2-q}$, Fubini's theorem and \eqref{3.14} give
\[
\begin{aligned}
 \var(E)
 =\frac4{\pi^2}\int_0^\infty\int_0^\infty\mathrm{Cov}(\zeta_v,\zeta_u)\,\dd u\,\dd v=2^{2-q}\int_{[0,\infty)}\bigg(\frac{2x}\pi\int_0^\infty\frac{vm_\rho(v)}{v^2+x^2}\,\dd v\bigg)^2\,\nu(\dd x)\,.
\end{aligned}
\]
Note that for $x>0$ and $E\geq0$, we have
\[
 \int_0^\infty\frac{v^2}{(v^2+E^2)(v^2+x^2)}\,\dd v=\frac{\pi}{2(E+x)}\,,\quad
 \frac{2x}\pi\int_0^\infty\frac{vm_\rho(v)}{v^2+x^2}\,\dd v=\int_{[0,\infty)}\frac{x}{E+x}\,\rho(\dd E)=1-\Lambda(x)\,.
\]
Since $\nu(\{0\})=0$ by \eqref{3.6}, this proves \eqref{3.15} and finishes the proof.
\end{proof}

\section{The spectral edge for growing \texorpdfstring{$q$}{q}} \label{sec4}
In this section, we prove Theorem \ref{mainthm1}\,(ii). Assume $4\leq q\leq N/2$, and set
\[
	\alpha\deq\frac{q^2}{N}\quad\mbox{and}\quad\kappa_N\deq\frac{2}{\sqrt{1-\e^{-2\alpha}}}\,.
\]
We first state the upper bound.

\begin{proposition} \label{prop4.1}
	If $q\to\infty$ and $q\leq N/2$, then
	\begin{equation} \label{4.1}
		\limsup_{N\to\infty}\frac{\bb E\|\mathcal H\|}{\kappa_N}\leq 1\,.
	\end{equation}
\end{proposition}

\begin{proof}[Proof of Theorem \ref{mainthm1}\,(ii)]
By Proposition \ref{prop4.1}, we have $\limsup_N \bb E\lambda_1/\kappa_N\leq 1$. For the lower bound, $\trn\e^{\beta H}\leq \e^{\beta d\lambda_1}$ gives $\bb E\lambda_1\geq Np_{\rm que}(\beta)/(\beta d)$ for all $\beta>0$. Passing to a subsequence, we can assume $\alpha=q^2/N\to \alpha_*\in [0,\infty]$.

If $\alpha_*=0$, so that $1\ll q\ll \sqrt{N}$, applying this bound at inverse temperature $b\beta$ with fixed $\beta>0$ gives
\[
	\frac{\bb E\lambda_1}{\kappa_N}\geq \frac{q^2 p_{\rm que}(b\beta)}{\beta}\frac{N}{q^2 bd\kappa_N}\,,\quad \mbox{where}\quad \frac{q^2 bd\kappa_N}{N}=\sqrt{\frac{a\alpha}{2}}\kappa_N\to 1\,.
\]
Here $a\to1$ and $1-\e^{-2\alpha}\sim2\alpha$. By Theorem \ref{mainthm2}\,(ii), we have $\liminf_N \bb E\lambda_1/\kappa_N\geq f(\beta)/\beta$. Taking $v=\pi/2$ in \eqref{bbb} gives $f(\beta)\geq\beta-\pi^2/4$, so letting $\beta\to\infty$ yields $\liminf_N \bb E\lambda_1/\kappa_N\geq1$.

If $\alpha_*\in(0,\infty]$, the global law of Feng--Tian--Wei \cite[Theorem 1]{FengTianWei1} gives
\[
 \liminf_{N\to\infty}\lambda_1\geq\frac2{\sqrt{1-\e^{-2\alpha_*}}}
 \quad\mbox{almost surely}\,.
\]
Since $\trn\mathcal H=0$, we have $\lambda_1\geq0$. Fatou's lemma and $\kappa_N\to2/\sqrt{1-\e^{-2\alpha_*}}$ therefore imply $\liminf_N\bb E\lambda_1/\kappa_N\geq1$. Thus $\lim_N\bb E\lambda_1/\kappa_N=1$.

Finally, since $q\leq N/2$ and $\kappa_N\geq 2$, for any fixed $\ee>0$, \eqref{2.4} gives
\[
	\sum_N \bb P\big(|\lambda_1-\bb E\lambda_1|>\ee\kappa_N\big)\leq \sum_N 2\e^{-N\ee^2}<\infty\,,
\]
and the Borel--Cantelli lemma implies $\lambda_1/\kappa_N\to 1$ almost surely. This finishes the proof.
\end{proof}

It remains to prove Proposition \ref{prop4.1}. Set $M\deq\binom Nq$. Define the $M\times M$ matrices $Q$ and $\mathcal J$, indexed by $q$-subsets $A,B\subset\qq{N}$, through
\[
	Q_{AB}\deq \frac{(-1)^{|A\cap B|}}{M}\quad \mbox{and} \quad \mathcal J_{AB}\deq \mathbf{1}_{|A\cap B|=q-1}\,.
\]
Let $\Pi_0$ be the orthogonal projection onto constant vectors in $\bb C^M$. The proof of Proposition \ref{prop4.1} is based on the following lemma, whose proof will be given in Section \ref{sec4Gaussian}.

\begin{lemma} \label{lem4.5}
	Suppose that
	\[
		Q\leq \sigma\Pi_0+\tau\mathcal J+\upsilon I\,,\quad 0\leq \sigma<1\,,\quad \tau,\upsilon\geq 0\,.
	\]
	Set $u\deq \tau N^2/4$ and $\ee\deq \upsilon/(1-\sigma)$. If $1+u-\ee\geq 0$, then we have
	\begin{equation} \label{4.8}
		\bb E\|\mathcal H\|\leq 2\sqrt{\frac{1+u-\ee}{1-\sigma}}+\sqrt{2\ee(N/2+1)\log 2}\,.
	\end{equation}
\end{lemma}

\subsection{The eigenvalues of \texorpdfstring{$Q$ and $\mathcal J$}{Q and J}} \label{sec4.1}
To estimate the eigenvalues of $Q$, we first bound expectations of $(-1)^{|A\cap T|}$ for uniform subsets $A$.
For integers $n\geq 1$, $0\leq m\leq n/2$, and $0\leq k\leq n$, let $A$ be a uniform $m$-subset of $\qq{n}$, and fix $T\subset\qq{n}$ with $|T|=k$. Set $\chi_{n,m}(k)\deq \bb E(-1)^{|A\cap T|}$ and $k_0\deq \min(k,n-k)$.

\begin{lemma} \label{lem2.10}
	We have
	\begin{equation} \label{2.42}
		\chi_{n,m}(k)=\sum_{j=0}^{\min(m,k)}\frac{(-2)^j(m)_j(k)_j}{j!(n)_j}
	\end{equation}
	and
	\begin{equation} \label{2.43}
		|\chi_{n,m}(k)|\leq \exp\bigg(-\frac{k_0}{4n}\min\Big(m,\Big\lfloor\frac{k_0}{4}\Big\rfloor\Big)\bigg)\,.
	\end{equation}
\end{lemma}
\begin{proof}
The identity \eqref{2.42} follows from $(-1)^{|A\cap T|}=(1-2)^{|A\cap T|}$ and $\bb E\binom{|A\cap T|}{j}=\binom{k}{j}(m)_j/(n)_j$. To prove \eqref{2.43}, we generate $A$ by first sampling $2m$ distinct elements $x_1,y_1,...,x_m,y_m\in \qq{n}$ uniformly at random, and then independently picking one element from each pair $\{x_j,y_j\}$ with probability $1/2$. Conditioned on $x_1,y_1,...,x_m,y_m$, we have
\[
	\bb E\big((-1)^{|A\cap T|}\mid x_1,y_1,...,x_m,y_m\big)=\prod_{j=1}^m \frac{(-1)^{\mathbf{1}_{x_j\in T}}+(-1)^{\mathbf{1}_{y_j\in T}}}{2}\,.
\]
Note that the $j$-th factor on the RHS of the above vanishes whenever $|\{x_j,y_j\}\cap T|=1$. Since replacing $T$ by $T^c$ only multiplies $(-1)^{|A\cap T|}$ by $(-1)^m$, we may assume W.L.O.G.\,that $k_0=k\leq n/2$. For $j\leq \min(m,\floor{k/4})$, conditioned on $x_1,y_1,...,x_{j-1},y_{j-1}$, there are at least $k/2$ elements left in $T$ and at least $n/4$ elements left in $T^c$. Thus the conditional probability of $|\{x_j,y_j\}\cap T|=1$ is at least $2(k/2)(n/4)/n^2=k_0/(4n)$. Conditioning successively on the sampled pairs gives
\[
\begin{aligned}
 |\chi_{n,m}(k)|
 &\leq\bb P\big(|\{x_j,y_j\}\cap T|\neq1\ \mbox{for all }1\leq j\leq\min(m,\floor{k_0/4})\big)\\
 &\leq\Big(1-\frac{k_0}{4n}\Big)^{\min(m,\floor{k_0/4})}\,,
\end{aligned}
\]
which proves \eqref{2.43}.
\end{proof}

For $0\leq k\leq q$ and $T\subset\qq{N}$ with $|T|=k$, set $f_T(A)\deq \mathbf{1}_{T\subset A}$, $\mathcal P_k\deq \mathrm{span}\{f_T:|T|=k\}$, and $\mathcal P_{-1}\deq \{0\}$. Define $\mathcal V_k\deq \mathcal P_k\cap \mathcal P_{k-1}^\perp$, and let $\Pi_k$ be the orthogonal projection onto $\mathcal V_k$ under the inner product $\langle f,g\rangle\deq\sum_{|A|=q}\ol{f(A)}g(A)$. For $n\geq2m\geq0$ and $n\geq1$, let $A,B$ be independent uniform $m$-subsets of $\qq n$, and abbreviate $\chi(n,m)\deq\bb E(-1)^{|A\cap B|}=\chi_{n,m}(m)$; see \eqref{2.42}. Set $\chi(0,0)\deq1$.

\begin{lemma} \label{lem4.2}
	We have $\bb C^M=\bigoplus_{k=0}^q \mathcal V_k$. On $\mathcal V_k$, $Q$ and $\mathcal J$ act as scalars with eigenvalues
	\begin{equation} \label{4.2}
		\theta_k=\frac{(-2)^k\binom{N-2k}{q-k}}{M}\chi(N-2k,q-k)\,,\quad j_k=q(N-q)-k(N-k+1)\,,\quad 0\leq k\leq q\,.
	\end{equation}
	In addition, we have
	\[
		2^k\frac{\binom{N-2k}{q-k}}{M}=\prod_{r=0}^{k-1}\frac{2(q-r)(N-q-r)}{(N-2r)(N-2r-1)}\leq \min\big((4q/N)^k,1,(2/3)^{k-1}\big)\,.
	\]
\end{lemma}
\begin{proof}
We use the Johnson eigenspaces from \cite[Fact 3.1]{BKM26}. The sum of the eigenspaces indexed by $0,...,k$ is $\mathcal P_k$, since
\[
 \operatorname{Ran}\bigg(\binom{|A\cap B|}{k}\bigg)_{A,B}
 =\operatorname{Ran}\sum_{|T|=k}f_Tf_T^*=\mathcal P_k\,.
\]
This gives the decomposition into $\mathcal V_k$. Applying the same fact to $\mathcal J+qI=(\binom{|A\cap B|}{q-1})_{A,B}$ gives $j_k$. For $Q$, \cite[(27)]{BKM26} gives
\[
 \theta_k=\frac{(-2)^k}{M}\sum_{l=0}^{q-k}(-2)^l
 \binom{q-k}{l}\binom{N-2k-l}{q-k-l}\,.
\]
By \eqref{2.42}, this is the expression in \eqref{4.2}; when $N=2k$ and $q=k$, the sum is $1$ and we use $\chi(0,0)=1$. Finally, in the product for $2^k\binom{N-2k}{q-k}/M$, each factor is bounded by $4q/N$, and is also bounded by $2/3$ when $N-2r>2$ and by $1$ when $N-2r=2$. This finishes the proof.
\end{proof}

\begin{lemma} \label{lem4.3}
	For $n\geq 2m\geq 0$ and $n\geq 1$, we have
	\begin{equation} \label{4.3}
		|\chi(n,m)|\leq 2\e^{-m^2/(32n)}\,.
	\end{equation}
	In addition, if $m/n\leq 1/8$, then
	\begin{equation} \label{4.4}
		0<\chi(n,m)\leq \e^{-2m^2/n}\,.
	\end{equation}
\end{lemma}
\begin{proof}
Taking $k=m$ in \eqref{2.43} gives \eqref{4.3}.

To prove \eqref{4.4}, let $A$ be a uniform $m$-subset of $\qq{n}$, and set $\chi_r\deq \chi_{n,m}(r)=\bb E(-1)^{|A\cap T|}$ for $0\leq r\leq m$, where $T\subset\qq{n}$ is a fixed subset of size $r$. Since $\sum_{x=1}^n (-1)^{\mathbf{1}_{x\in A}}=n-2m$, multiplying by $(-1)^{|A\cap T|}$, splitting the sum into $x\notin T$ and $x\in T$, and taking expectations give
\[
	(n-r)\chi_{r+1}+r\chi_{r-1}=(n-2m)\chi_r\,,\quad \chi_0=1\,,\quad \chi_1=1-\frac{2m}{n}\,.
\]
For $n\geq 8m$ and $1\leq r<m$, since $n-2m-r/(1-4m/n)\geq (n-r)(1-4m/n)$, induction on $r$ gives
\[
	1-\frac{4m}{n}\leq \frac{\chi_{r+1}}{\chi_r}=\frac{n-2m-r/(\chi_r/\chi_{r-1})}{n-r}\leq 1-\frac{2m}{n}\,.
\]
Multiplying over $0\leq r<m$ yields $0<\chi_m\leq (1-2m/n)^m\leq \e^{-2m^2/n}$. This finishes the proof.
\end{proof}

\subsection{The upper bound} \label{sec4.2}
We now combine Lemmas \ref{lem4.5}, \ref{lem4.2}, and \ref{lem4.3} to prove Proposition \ref{prop4.1}.

\begin{proof}[Proof of Proposition \ref{prop4.1}]
First suppose $q\leq N/16$. By \eqref{4.2} and \eqref{4.4}, $\theta_k<0$ for odd $k$, and $0<\theta_0\leq \e^{-2\alpha}$. We bound the $k=2$ term by a multiple of $\mathcal J+qI$, and bound the terms with $k\geq4$ by a multiple of $I$. To this end, for large $q$, set
\[
	\tau\deq \frac{\theta_2}{q(N-q+1)-2(N-1)}\,,\quad \sigma\deq \theta_0\,,\quad \upsilon\deq q\tau+\max_{4\leq k\leq q}(\theta_k)_+\,,\quad \mbox{and} \quad u\deq \frac{\tau N^2}{4}\,.
\]
By \eqref{4.2}, we have
\[
	\mathcal J+qI=\sum_{k=0}^q (q-k)(N-q-k+1)\Pi_k\geq \big(q(N-q+1)-2(N-1)\big)\Pi_2\geq 0\,,
\]
and thus
\[
	Q\leq \sigma\Pi_0+\theta_2\Pi_2+\max_{4\leq k\leq q}(\theta_k)_+ I\leq \sigma\Pi_0+\tau\mathcal J+\upsilon I\,.
\]
To apply \eqref{4.8}, we estimate $u$ and $\upsilon/(1-\sigma)$. For $k\leq q/2$, we have $(q-k)/(N-2k)\leq q/N\leq1/16$ and $(q-k)^2/(N-2k)\geq q^2/(4N)$. Thus \eqref{4.4} gives $|\chi(N-2k,q-k)|\leq\e^{-\alpha/2}$. For $k>q/2$, in the product of Lemma \ref{lem4.2}, we bound the first four factors by $4q/N$ and the remaining factors by $2/3$. Writing $t_q\deq (2/3)^{q/2-4}$, we get
\[
	\theta_2\leq \frac{16q^2}{N^2}\e^{-\alpha/2}\quad \mbox{and} \quad \max_{4\leq k\leq q}(\theta_k)_+\leq \frac{256q^4}{N^4}(\e^{-\alpha/2}+t_q)\,.
\]
As a result,
\[
	u\leq C\frac{q}{N}\e^{-\alpha/2}\leq CN^{-1/2}\quad \mbox{and} \quad \upsilon\leq C\Big(\frac{q^2}{N^3}\e^{-\alpha/2}+\frac{q^4}{N^4}(\e^{-\alpha/2}+t_q)\Big)\,.
\]
Setting $\ee\deq\upsilon/(1-\sigma)$ and using $1-\sigma\geq c\min(\alpha,1)$, we obtain
\[
	\ee(N/2+1)\log 2\leq C\frac{\alpha\e^{-\alpha/2}+\alpha^2(\e^{-\alpha/2}+t_q)}{N\min(\alpha,1)}\leq \frac{C}{N}+Cqt_q\to 0\,.
\]
In particular, $\ee\to0$, so $1+u-\ee>0$ for all large $N$. Since $\sigma\leq\e^{-2\alpha}$, applying \eqref{4.8} gives $\bb E\|\mathcal H\|/\kappa_N\leq\sqrt{1+u}+\kappa_N^{-1}\sqrt{2\ee(N/2+1)\log 2}=1+o(1)$.

Next suppose $q>N/16$. For $k\leq q/2$, Lemma \ref{lem4.3} gives $|\chi(N-2k,q-k)|\leq2\e^{-q^2/(128N)}$, and the factor $2^k\binom{N-2k}{q-k}/M$ in Lemma \ref{lem4.2} is at most $1$. For $k>q/2$, that factor is at most $(2/3)^{q/2-1}$. Hence
\begin{equation} \label{4.11}
	N\max_{0\leq k\leq q}|\theta_k|\leq 2N\e^{-q^2/(128N)}+N(2/3)^{q/2-1}\to 0\,.
\end{equation}
Taking $\tau=u=0$, $\sigma=(\theta_0)_+$, and $\upsilon=\max_{1\leq k\leq q}(\theta_k)_+$ in Lemma \ref{lem4.5}, we get $\sigma\to 0$ and $\ee(N/2+1)\log 2\to 0$, so \eqref{4.8} gives $\bb E\|\mathcal H\|\leq 2+o(1)=(1+o(1))\kappa_N$. This proves \eqref{4.1}.
\end{proof}

\subsection{Proof of Lemma \ref{lem4.5}} \label{sec4Gaussian}
We shall prove \eqref{4.8} by bounding $\bb E\trn\mathcal H^{2k}$ and summing the exponential series.
For polynomial matrix functions of $J=(J_A)_{|A|=q}$, define $\langle D,U\rangle\deq\bb E\trn(D^*U)$ and $\|D\|_2\deq\langle D,D\rangle^{1/2}$. Following \cite{BKM26}, we set
\[
	\partial_A\deq \frac{\partial}{\partial J_A}\,,\quad \mathcal A\deq \frac{1}{\sqrt{M}}\sum_A \Psi_A\partial_A\,,\quad \mathcal A^*\deq \frac{1}{\sqrt{M}}\sum_A \Psi_A(J_A-\partial_A)\,,\quad \mathcal N\deq \sum_A (J_A-\partial_A)\partial_A\,.
\]
Here every sum is over $q$-subsets of $\qq N$. Stein's Lemma gives
\[
 \langle\partial_A D,U\rangle=\langle D,(J_A-\partial_A)U\rangle\,,
\]
so $\mathcal A^*$ is the adjoint of $\mathcal A$. Also $(\mathcal A+\mathcal A^*)D=\mathcal HD$. Write $D=\sum_{h\geq0}D_h$, where $D_h$ is the sum of its multivariate Hermite terms of total degree $h$. These components are orthogonal, and
\begin{equation} \label{4.5a}
	\mathcal N D_h=hD_h\,,\quad [\mathcal N,\mathcal A]=-\mathcal A\,,\quad \mbox{and} \quad [\mathcal N,\mathcal A^*]=\mathcal A^*\,.
\end{equation}
To apply $Q\leq \sigma\Pi_0+\tau\mathcal J+\upsilon I$, set $W_A\deq \Psi_A\partial_A D$ and $\b W\deq (W_A)_{|A|=q}$, and write
\[
	\|\b W\|_2^2\deq \sum_A \|W_A\|_2^2\,,\quad \langle \b W,Q\b W\rangle\deq \sum_{A,B}Q_{AB}\langle W_A,W_B\rangle\,.
\]
Since $(\Pi_0\b W)_A=M^{-1}\sum_BW_B$, Stein's Lemma and \eqref{1.2} give
\begin{equation} \label{4.5}
	\|\mathcal A^*D\|_2^2=\|D\|_2^2+\langle \b W,Q\b W\rangle\,,\quad \|\mathcal A D\|_2^2=\|\Pi_0\b W\|_2^2\,,\quad \mbox{and} \quad \|\b W\|_2^2=\langle D,\mathcal N D\rangle\,.
\end{equation}
For the first relation, we have
\[
 \mathcal A\mathcal A^*D=D+\frac1M\sum_{A,B}\Psi_A\Psi_B(J_B-\partial_B)\partial_A D\,,
\]
where the first term is $M^{-1}\sum_A\Psi_A^2D=D$. Thus Stein's Lemma gives
\[
	\|\mathcal A^*D\|_2^2-\|D\|_2^2=\frac{1}{M}\sum_{A,B}\langle \partial_B D,\Psi_A\Psi_B\partial_A D\rangle=\frac{1}{M}\sum_{A,B}(-1)^{|A\cap B|}\langle W_A,W_B\rangle\,.
\]
In the last equality we interchanged $A,B$ and used \eqref{1.2}.

We next bound $\langle \b W,\mathcal J\b W\rangle$ using commutators with $\psi_i\psi_j$. For this estimate we use the rotation relation satisfied by polynomials in $\mathcal H$. For $1\leq i<j\leq N$, set $D_{ij}\deq -\ii\psi_i\psi_j/2$ and $U_{ij}(v)\deq \e^{\ii vD_{ij}}=\cos(v/2)I+\sin(v/2)\psi_i\psi_j$. Note that
\[
U_{ij}(v)\psi_iU_{ij}(v)^*=\cos(v)\psi_i-\sin(v)\psi_j\,,\quad U_{ij}(v)\psi_jU_{ij}(v)^*=\sin(v)\psi_i+\cos(v)\psi_j\,,
\]
and $U_{ij}(v)\psi_kU_{ij}(v)^*=\psi_k$ for $k\notin\{i,j\}$. Conjugation by $U_{ij}(v)$ leaves $\Psi_A$ unchanged if $A$ contains both or neither of $i,j$. Otherwise, $U_{ij}(v)\Psi_AU_{ij}(v)^*$ is a real linear combination of $\Psi_A$ and $\Psi_{A\triangle\{i,j\}}$. Since $\trn(\Psi_A\Psi_B)=\delta_{AB}$, we can define the orthogonal matrix $V_{ij}(v)$ on $\bb R^M$ by $U_{ij}(v)\Psi_A U_{ij}(v)^*=\sum_B V_{ij}(v)_{BA}\Psi_B$. Let $\mathcal X$ be the space of polynomial matrix functions $D(J)$ satisfying
\begin{equation} \label{4.6}
	D\big(V_{ij}(v)J\big)=U_{ij}(v)D(J)U_{ij}(v)^*\quad \mbox{for all}\ 1\leq i<j\leq N\,,\ v\in \bb R\,.
\end{equation}
Since $\mathcal H(V_{ij}(v)J)=U_{ij}(v)\mathcal H(J)U_{ij}(v)^*$, the space $\mathcal X$ contains $I$ and all polynomials in $\mathcal H$. For every orthogonal matrix $V$, the chain rule gives $\mathcal N(D\circ V)=(\mathcal ND)\circ V$, so changing the Gaussian variables by $V$ preserves each Hermite degree. Taking the degree-$h$ component in \eqref{4.6} therefore shows that $D_h\in\mathcal X$ whenever $D\in\mathcal X$.

\begin{lemma} \label{lem4.4}
	For every $D\in \mathcal X$, we have $\mathcal A D, \mathcal A^*D\in \mathcal X$, and
	\begin{equation} \label{4.7}
		\langle \b W,\mathcal J\b W\rangle=\sum_{1\leq i<j\leq N}\|[D_{ij},D]\|_2^2\leq \frac{N^2}{4}\|D\|_2^2\,.
	\end{equation}
\end{lemma}
\begin{proof}
Fix $i<j$, and abbreviate $V\equiv V_{ij}(v)$ and $U\equiv U_{ij}(v)$. Differentiating \eqref{4.6} in $J$ gives $(\partial_A D)(VJ)=\sum_B V_{AB}U(\partial_B D)(J)U^*$. Together with $\sum_A V_{AB}\Psi_A=U\Psi_B U^*$, this implies $\mathcal A D\in \mathcal X$ and $\mathcal A^*D=(\mathcal H-\mathcal A)D\in \mathcal X$.

For the proof of \eqref{4.7}, set $X\deq V_{ij}'(0)$ and $\mathcal R_{ij}\deq \sum_{A,B}X_{AB}J_B\partial_A$. From the action of $U_{ij}(v)$ on $\psi_k$, we see that $X_{AB}$ is nonzero only when $B=A\triangle\{i,j\}$, where for $i\in A$ and $j\notin A$,
\[
	X_{AB}=(-1)^{|A\cap\{i+1,...,j-1\}|}=-X_{BA}\quad \mbox{and} \quad \Psi_A\Psi_B=2\ii X_{AB}D_{ij}\,.
\]
Differentiating \eqref{4.6} at $v=0$ gives $\mathcal R_{ij}D=\ii[D_{ij},D]$. Since $X$ is antisymmetric, we have $\sum_{A,B}X_{AB}\partial_B\partial_A D=0$. Thus Stein's Lemma gives
\[
 \langle D,D_{ij}\mathcal R_{ij}D\rangle
 =\sum_{A,B}X_{AB}\langle\partial_B D,D_{ij}\partial_A D\rangle\,.
\]
Every pair $(A,B)$ with $|A\cap B|=q-1$ has $A\triangle B=\{i,j\}$ for a unique pair $i<j$. Together with $D_{ij}^2=I/4$, we get
\[
\begin{aligned}
	\langle \b W,\mathcal J\b W\rangle=&\,2\ii\sum_{i<j}\sum_{A,B}X_{AB}\langle \partial_A D,D_{ij}\partial_B D\rangle=-2\ii\sum_{i<j}\langle D,D_{ij}\mathcal R_{ij}D\rangle\\
	=&\,2\sum_{i<j}\langle D,D_{ij}[D_{ij},D]\rangle=\sum_{i<j}\|[D_{ij},D]\|_2^2\,.
\end{aligned}
\]
Finally, expanding $\big\|\sum_{i=1}^N \psi_i D\psi_i\big\|_2^2$ via \eqref{2.2} gives
\[
	4\sum_{i<j}\|[D_{ij},D]\|_2^2=2\sum_{i<j}\big(\|D\|_2^2-\re\langle D,\psi_i\psi_j D\psi_j\psi_i\rangle\big)=N^2\|D\|_2^2-\Big\|\sum_{i=1}^N \psi_i D\psi_i\Big\|_2^2\leq N^2\|D\|_2^2\,,
\]
which proves \eqref{4.7}.
\end{proof}

The comparison in \cite{BKM26} uses the bound
\[
 \langle\b W,(Q-\theta_0\Pi_0)\b W\rangle\leq (\max_{1\leq k\leq q}\theta_k)_+\langle D,\mathcal ND\rangle\,.
\]
Here, by \eqref{4.7}, the term $\tau\mathcal J$ contributes at most $(\tau N^2/4)\|D\|_2^2$, independently of the Hermite degree. The remaining error is $\upsilon\langle D,\mathcal ND\rangle$.

\begin{proof}[Proof of Lemma \ref{lem4.5}]
Recall $u=\tau N^2/4$ and $\ee=\upsilon/(1-\sigma)$. By \eqref{4.5}, \eqref{4.7}, and $Q\leq \sigma\Pi_0+\tau\mathcal J+\upsilon I$, we have for all $D\in \mathcal X$ that
\[
	\|\mathcal A^*D\|_2^2\leq (1+u)\|D\|_2^2+\sigma\|\mathcal A D\|_2^2+\upsilon\langle D,\mathcal N D\rangle\,.
\]
To turn this inequality into bounds on moments, set $\kappa\deq 2\sqrt{(1+u-\ee)/(1-\sigma)}$ and $\mathcal X_h\deq \{D\in \mathcal X:\mathcal N D=hD\}$. By Lemma \ref{lem4.4} and \eqref{4.5a}, the maps $\mathcal A:\mathcal X_h\to\mathcal X_{h-1}$ and $\mathcal A^*:\mathcal X_{h-1}\to\mathcal X_h$ are adjoints. Thus we have
\[
	\|\mathcal A|_{\mathcal X_h}\|^2=\|\mathcal A^*|_{\mathcal X_{h-1}}\|^2\leq 1+u+\sigma\|\mathcal A|_{\mathcal X_{h-1}}\|^2+\upsilon(h-1)\,,\quad h\geq 1\,.
\]
Since $\mathcal A|_{\mathcal X_0}=0$ and $1+u+\sigma(\kappa^2/4+\ee(h-1))+\upsilon(h-1)=\kappa^2/4+\ee h$, induction on $h$ gives, for all $D\in \mathcal X_h$,
\begin{equation} \label{4.9}
	\|\mathcal A D\|_2^2\leq (\kappa^2/4+\ee h)\|D\|_2^2\quad \mbox{and} \quad \|\mathcal A^*D\|_2^2\leq \big(\kappa^2/4+\ee(h+1)\big)\|D\|_2^2\,.
\end{equation}

Fix $k\geq0$ and expand $\bb E\trn\mathcal H^{2k}=\langle I,(\mathcal A+\mathcal A^*)^{2k}I\rangle$. For a nonzero term, after $j$ factors have acted on $I$ from the right, let $h_j$ be the Hermite degree. Since $\mathcal A\mathcal X_0=\{0\}$ and different Hermite degrees are orthogonal, only paths with $h_0=h_{2k}=0$, $h_j\geq0$, and $|h_j-h_{j-1}|=1$ contribute. Cauchy--Schwarz and \eqref{4.9} bound the absolute value of each term by a product with one factor $(\kappa^2/4+\ee h)^{1/2}$ for each step between levels $h-1$ and $h$. Since there are equally many up-steps and down-steps between these levels, we get
\[
	\bb E\trn \mathcal H^{2k}\leq \sum_{\substack{h_0=h_{2k}=0, h_j\geq 0\\ |h_j-h_{j-1}|=1}}\prod_{j: h_j=h_{j-1}-1}\big(\kappa^2/4+\ee h_{j-1}\big)\,.
\]
We expand the product on the RHS. At each down-step $j$, the factor $\ee h_{j-1}$ allows us to pair $j$ with any of the $h_{j-1}$ unmatched up-steps. For the factor $\kappa^2/4$, we pair with the last unmatched up-step. Two pairs of the latter type cannot satisfy $a_1<a_2<b_1<b_2$, since $a_2$ would still be unmatched at $b_1$.

Mark the pairs that use the factor $\ee h_{j-1}$. The marked pairs and the non-crossing pairing on the remaining endpoints determine the path and all pairing choices. If there are $r$ marked pairs, choosing their endpoints, pairing these endpoints, and choosing the up-steps on the complement give at most
\[
 \binom{2k}{2r}(2r-1)!!\binom{2(k-r)}{k-r}
 \leq\binom{2k}{2r}(2r-1)!!4^{k-r}
\]
possibilities. Thus, for a standard Gaussian $z$, we have
\begin{equation} \label{4.10}
	\bb E\trn \mathcal H^{2k}\leq \sum_{r=0}^k \binom{2k}{2r}(2r-1)!!\ee^r\kappa^{2(k-r)}=\bb E(\kappa+\sqrt{\ee}z)^{2k}\,.
\end{equation}
Here $(-1)!!=1$. Since all odd moments vanish by $J\mapsto-J$, summing over $k\geq 0$ gives, for all $t>0$,
\[
	\bb E\trn \e^{t\mathcal H}\leq \bb E\cosh\big(t(\kappa+\sqrt{\ee}z)\big)=\cosh(\kappa t)\e^{\ee t^2/2}\leq \e^{\kappa t+\ee t^2/2}\,.
\]
Since $\e^{t\|\mathcal H\|}\leq\tr(\e^{t\mathcal H}+\e^{-t\mathcal H})$, Jensen's inequality and symmetry $J\mapsto-J$ give, for $t>0$,
\[
	\bb E\|\mathcal H\|\leq \frac{1}{t}\log\bb E\tr(\e^{t\mathcal H}+\e^{-t\mathcal H})=\frac{(N/2+1)\log 2+\log\bb E\trn\e^{t\mathcal H}}{t}\leq \kappa+\frac{(N/2+1)\log 2}{t}+\frac{\ee t}{2}\,.
\]
Optimizing over $t>0$ gives \eqref{4.8}.
\end{proof}

\section{Growing inverse temperature and complementary order} \label{sec5}
In this section, we prove Corollary \ref{mainthm3} and extend our results to the complementary interaction order $N-q$.

\begin{proof}[Proof of Corollary \ref{mainthm3}]
For part (i), recall that $-\beta f_{\rm que}(\beta)=p_{\rm que}(\beta)+(\log2)/2$ for $\beta>0$. By \eqref{3.1} and $\lambda_{\max}(H)=d\lambda_1$, we have
\begin{equation} \label{5.1}
\begin{aligned}
 \frac{d\lambda_1}{N}-\frac{\log2}{2\beta_N}
 \leq\frac1{N\beta_N}\log\trn\e^{\beta_N H}
 \leq\frac{d\lambda_1}{N}\,,\quad
 0\leq-f_{\rm que}(\beta_N)-\frac dN\bb E\lambda_1
 \leq\frac{\log2}{2\beta_N}\,.
\end{aligned}
\end{equation}
By Theorem \ref{mainthm1}\,(i) and $a\to1$, we have almost surely
\[
 \frac{d\lambda_1}{N}
 =\sqrt{\frac{a}{q2^q}}\,\frac{\lambda_1}{\sqrt N}
 \to e_{{\rm SD},q}\,.
\]
By Lemma \ref{lem3.1}, the same convergence holds for $d\bb E\lambda_1/N$. Since $\beta_N\to\infty$, \eqref{5.1} proves part (i).

For parts (ii) and (iii), fix $\beta_0>0$. The function $t\mapsto\log\trn\e^{t\mathcal H}$ is convex and vanishes at $0$. Since $\beta_N\geq\beta_0$ for all large $N$, we have, in parts (ii) and (iii), respectively,
\[
\begin{aligned}
 \frac{q^2p_{\rm que}(b\beta_0)}{\beta_0}
 \leq\frac{q^2p_{\rm que}(b\beta_N)}{\beta_N}
 \leq\frac{q^2bd}{N}\bb E\lambda_1\,,\quad
 \frac{Np_{\rm que}(\beta_0/d)}{\beta_0}
 \leq\frac{Np_{\rm que}(\beta_N/d)}{\beta_N}
 \leq\bb E\lambda_1\,.
\end{aligned}
\]
By Theorem \ref{mainthm2}\,(ii), the first left-hand side converges to $f(\beta_0)/\beta_0$. For the second, the global law and the uniform moment bounds in \eqref{2.45} give convergence of each fixed moment. Moreover,
\[
 \sum_{k>K}\frac{\beta_0^{2k}}{(2k)!}\bb E\trn\mathcal H^{2k}
 \leq\sum_{k>K}\frac{(\beta_0^2/2)^k}{k!}\to0
 \quad\mbox{as}\ K\to\infty\,,
\]
uniformly in $N$. Summing the exponential series and using \eqref{2.5}, we get $Np_{\rm que}(\beta_0/d)\to L_\fs(\beta_0)$ as in Remark \ref{rem1.2}\,(ii).

For part (ii), $q^2bd/N=\sqrt{a\alpha/2}$ with $\alpha=q^2/N\to0$, and $\sqrt{a\alpha/2}\,\kappa_N\to1$. Theorem \ref{mainthm1} and \eqref{2.4} therefore show that the right-hand sides converge to $1$ and $2/\sqrt{1-\fs}$, respectively. Letting $\beta_0\to\infty$ and using
\[
 \lim_{\beta_0\to\infty}\frac{f(\beta_0)}{\beta_0}=1
 \quad\mbox{and}\quad
 \lim_{\beta_0\to\infty}\frac{L_\fs(\beta_0)}{\beta_0}
 =\frac2{\sqrt{1-\fs}}\,,
\]
we conclude the proof.
\end{proof}

\subsection{Complementary order} \label{reflection}
Finally, we compare interaction orders $q$ and $N-q$. For $j\in\{q,N-q\}$, write $H_j=d_j\mathcal H_j$ and denote the corresponding pressures by $p_{{\rm ann},j}$ and $p_{{\rm que},j}$. Let $\lambda_{1,j}$ be the largest eigenvalue of $\mathcal H_j$. The explicit Majorana matrices give $\Gamma=\ii^{N/2}\psi_1\cdots\psi_N$.

\begin{lemma} \label{lem5.1}
For even $4\leq q\leq N/2$, we have $\mathcal H_{N-q}\overset{d}{=}\Gamma\mathcal H_q$, and hence
\begin{equation} \label{5.2}
 \mathcal H_{N-q}^2\overset{d}{=}\mathcal H_q^2\,,\quad \|\mathcal H_{N-q}\|\overset{d}{=}\|\mathcal H_q\|\,,\quad
 p_{{\rm ann},N-q}(\beta/d_{N-q})=p_{{\rm ann},q}(\beta/d_q)\,.
\end{equation}
In addition, for $j\in\{q,N-q\}$, both extreme eigenvalues of $\mathcal H_j$ satisfy \eqref{2.4} with $N-q+1$ unchanged. For every $\beta\geq0$, we also have
\begin{equation} \label{5.3}
\begin{aligned}
 0\leq p_{{\rm ann},j}(\beta/d_j)-p_{{\rm que},j}(\beta/d_j)&\leq\frac{\beta^2}{2N(N-q+1)}\,,\\
 |p_{{\rm que},N-q}(\beta/d_{N-q})-p_{{\rm que},q}(\beta/d_q)|&\leq\frac{\beta^2}{2N(N-q+1)}\,,
\end{aligned}
\end{equation}
and
\begin{equation} \label{5.4}
 |\bb E\lambda_{1,N-q}-\bb E\lambda_{1,q}|\leq(N-q+1)^{-1/2}\,.
\end{equation}
\end{lemma}
\begin{proof}
For each $q$-subset $B\subset\qq N$, set $\sigma_B\deq(-1)^{|\{(b,c)\in B\times B^c:b>c\}|}$. By \eqref{2.2}, we have
\[
 \Gamma^*=\Gamma\,,\quad \Gamma^2=I\,,\quad [\Gamma,\Psi_B]=0\,,\quad \Gamma\Psi_B=\sigma_B\Psi_{B^c}\,.
\]
Setting $J'_{B^c}=\sigma_BJ_B$ gives $\mathcal H_{N-q}(J')=\Gamma\mathcal H_q(J)$. Since $(J'_A)_{|A|=N-q}\overset{d}{=}(J_A)_{|A|=N-q}$, this gives $\mathcal H_{N-q}\overset{d}{=}\Gamma\mathcal H_q$ and the first two relations in \eqref{5.2}. By symmetry $J\mapsto-J$, we have
\[
 \bb E\trn\e^{\beta\mathcal H_{N-q}}=\bb E\trn\cosh(\beta\mathcal H_{N-q})
 =\bb E\trn\cosh(\beta\mathcal H_q)=\bb E\trn\e^{\beta\mathcal H_q}\,,
\]
which gives the last relation in \eqref{5.2}.

For each $T\subset\qq N$ with $|T|=q-1$, we have
\[
 \{\Gamma\Psi_{T\cup\{x\}},\Gamma\Psi_{T\cup\{y\}}\}=2\delta_{xy}I\,,\quad x,y\notin T\,.
\]
Let $W\geq0$ satisfy $\tr W=1$, and set $\phi(D)\deq\tr(WD)$. Since $\Gamma\Psi_{T\cup\{x\}}=\pm\Psi_{(T\cup\{x\})^c}$, the anticommuting estimate in the proof of Lemma \ref{lemma2.1} gives
\[
 q\sum_{|A|=N-q}\phi(\Psi_A)^2
 =\sum_{|T|=q-1}\sum_{x\notin T}\phi(\Gamma\Psi_{T\cup\{x\}})^2
 \leq\binom N{q-1}=\frac{q}{N-q+1}\binom Nq\,.
\]
Dividing by $q\binom Nq$ and applying the Gaussian concentration argument from that proof give \eqref{2.4} and the first bound in \eqref{5.3}. Combining the first bound in \eqref{5.3} with the last relation in \eqref{5.2} yields the second bound in \eqref{5.3}.

To prove \eqref{5.4}, let $\lambda_{-,j}$ be the smallest eigenvalue of $\mathcal H_j$. By symmetry, $\bb E\lambda_{-,j}=-\bb E\lambda_{1,j}$. Thus \eqref{2.4} and the triangle inequality give
\[
 0\leq\bb E\|\mathcal H_j\|-\bb E\lambda_{1,j}
 =\frac12\bb E|\lambda_{1,j}+\lambda_{-,j}|\leq(N-q+1)^{-1/2}\,.
\]
Together with $\bb E\|\mathcal H_{N-q}\|=\bb E\|\mathcal H_q\|$, this proves \eqref{5.4}.
\end{proof}

For the pressure limits in Theorem \ref{mainthm2} and Remark \ref{rem1.2}\,(ii), we apply \eqref{5.3} at $\beta d_q$, $\beta bd_q$, and $\beta$, respectively, and obtain
\[
\begin{aligned}
 |p_{{\rm que},N-q}(\beta d_q/d_{N-q})-p_{{\rm que},q}(\beta)|&\leq\frac{a\beta^2}{q2^{q+1}(N-q+1)}\to0\,,\\
 q^2|p_{{\rm que},N-q}(\beta bd_q/d_{N-q})-p_{{\rm que},q}(b\beta)|&\leq\frac{a\beta^2}{4(N-q+1)}\to0\,,\\
 N|p_{{\rm que},N-q}(\beta/d_{N-q})-p_{{\rm que},q}(\beta/d_q)|&\leq\frac{\beta^2}{2(N-q+1)}\to0\,.
\end{aligned}
\]
The annealed limits follow from the identity in \eqref{5.2}. Similarly, \eqref{5.4} and \eqref{2.4} yield the edge limits in Theorem \ref{mainthm1}. Finally, the proof of Corollary \ref{mainthm3} applies to $(d_q/d_{N-q})H_{N-q}$ in place of $H$, by the fixed-temperature pressure limits and the edge limits above.

\end{document}